\documentclass[
  aps,
  prx,
  twocolumn,
  superscriptaddress
]{revtex4-2}

\usepackage{microtype}
\usepackage{blindtext}
\usepackage{amsmath}
\usepackage{amsfonts}
\usepackage{amsthm}
\usepackage{physics}
\usepackage{comment}
\usepackage{makecell}
\usepackage{mathtools}
\usepackage{bm}
\usepackage{CJKutf8}

\usepackage[
    colorlinks=true,
    hypertexnames=false
]{hyperref}

\hypersetup{
    allcolors=blue
}

\usepackage{url}
\usepackage[all]{hypcap}

\usepackage[capitalise,compress,nameinlink]{cleveref}
\AddToHook{cmd/appendix/before}{
  \crefalias{section}{appendix}
}
\usepackage{enumitem}   
\usepackage{tikz} 
\usetikzlibrary{arrows, shapes.gates.logic.US, calc}
\usetikzlibrary{shapes}
\usetikzlibrary{plotmarks}
\usetikzlibrary{quantikz}
\usetikzlibrary{patterns}

\DeclareMathOperator{\poly}{poly}

\usepackage{xcolor}
\usepackage{bbm}

\newtheorem{theorem}{Theorem}

\newtheorem{fact}{Fact}
\newtheorem{lemma}{Lemma}
\newtheorem{definition}{Definition}

\crefname{protocol}{Protocol}{Protocols}
\Crefname{protocol}{Protocol}{Protocols}

\newlength{\protowidth}

\usepackage{algpseudocode}
\usepackage{algorithm}

\usepackage{graphicx}

\newcommand{\OC}{\mathcal{O}}

\let\originalleft\left
\let\originalright\right
\renewcommand{\left}{\mathopen{}\mathclose\bgroup\originalleft}
\renewcommand{\right}{\aftergroup\egroup\originalright}

\providecommand{\textcite}[1]{\citet{#1}}

\begin{document}
\title{Ordering-Aware Theory of Trotter Error}
\author{Zhenning~Liu~(\begin{CJK*}{UTF8}{gbsn}刘镇宁\end{CJK*})}
\email{zhenning@umd.edu}
\affiliation{Joint Center for Quantum Information and Computer Science,
NIST/University of Maryland, College Park, MD 20742, USA}
\affiliation{Department of Computer Science,
University of Maryland, College Park, MD 20742, USA}

\author{Zhi-Yuan~Wei~(\begin{CJK*}{UTF8}{gbsn}魏志远\end{CJK*})}
\email{zywei@umd.edu}
\affiliation{Joint Center for Quantum Information and Computer Science,
NIST/University of Maryland, College Park, MD 20742, USA}
\affiliation{Joint Quantum Institute,
NIST/University of Maryland, College Park, MD 20742, USA}

\author{Michael~Foss-Feig}
\affiliation{Quantinuum, 303 S Technology Ct, Broomfield, CO 80021, USA}

\author{Alexey~V.~Gorshkov}
\affiliation{Joint Center for Quantum Information and Computer Science,
NIST/University of Maryland, College Park, MD 20742, USA}
\affiliation{Joint Quantum Institute,
NIST/University of Maryland, College Park, MD 20742, USA}

\author{Andrew~M.~Childs}
\affiliation{Joint Center for Quantum Information and Computer Science,
NIST/University of Maryland, College Park, MD 20742, USA}
\affiliation{Department of Computer Science,
University of Maryland, College Park, MD 20742, USA}
\affiliation{Institute for Advanced Computer Studies,
University of Maryland, College Park, MD 20742, USA}

\begin{abstract}
Product formulas are a basic tool for simulating quantum many-body dynamics, but a fundamental degree of freedom in their application---the ordering of Hamiltonian terms---has remained largely unexplored. Here we systematically study such ordering effects, quantifying their impact on simulation accuracy and exploiting them to design more accurate product-formula simulations. We develop an efficiently computable, ordering-sensitive bound on the finite-step-size
Trotter error for geometrically local Hamiltonians for arbitrarily high-order product formulas. 
For translation-invariant one-dimensional chains, our bounds are close to exact second-, fourth-, and sixth-order Trotter errors in small-system tests, while on large-scale systems
they improve on previous higher-order analytical bounds by orders of magnitude.
We further use the local structure of the error bounds and estimates to optimize ordering efficiently, reducing the long-chain problem to a minimum-mean-weight-cycle problem on a weighted de Bruijn graph. For random translation-invariant chains, sequential (staircase-like) 
orderings consistently outperform standard brickwall circuits; when circuit depth is restricted, the optimizer instead favors wave-like parallel-sequential orderings, which preserve much of the accuracy advantage over brickwall circuits. These results establish ordering as a useful degree of freedom for both understanding and improving product-formula simulation.
\end{abstract}
\maketitle

\section{Introduction}
\label{sec:introduction}

\subsection{Background and motivation}
Simulating the real-time dynamics of quantum many-body systems is a central task in both quantum computation and computational physics
\cite{georgescu2014quantum,RevModPhys.83.863}.
Product formulas approximate an evolution
generated by $H=\sum_{j=1}^n H_j$ using a sequence of evolutions under its
individual terms $H_j$ \cite{trotter1959product,suzuki1990fractal}, an approach also referred to as Trotterization.
The simplicity, locality, and low implementation overhead of product formulas make them one of the most practically relevant methods for digital quantum simulation \cite{Lloyd1996,berry2007efficient,childs2019nearly,childs2021theory}.
The same decompositions are also fundamental to classical simulation of many-body dynamics.
In particular, they underlie time-evolving block decimation (TEBD) and related tensor-network algorithms, in which local evolution operators are successively applied to a matrix product state~\cite{Vidal2004,paeckel2019time}.
Product formulas are thus a common algorithmic primitive shared by quantum and classical simulation.

The error in Trotterization has been studied extensively in terms of the step size, the formula order, and the locality and commutator structure of the Hamiltonian
\cite{trotter1959product,suzuki1990fractal,berry2007efficient,Lloyd1996,childs2019nearly,childs2021theory}.
There is, however, another elementary choice that is required to specify a product-formula simulation: the ordering of the terms of the Hamiltonian, which determines the order in which local evolutions are applied 
[see~\cref{fig_illu}(a)]. 
Different orderings lead to different formulas,
and since the terms generally do not
commute, their Trotter errors need not be the same.
For example, for a one-dimensional chain, the even--odd ordering gives a shallow \emph{brickwall} circuit, whereas ordering the terms along the chain gives a \emph{sequential}, or staircase-like, circuit.
More generally, both geometries fit within a unified family of parallel-sequential circuits~\cite{wei2025state}, interpolating between maximal parallelism and a spatially sequential sweep.
In most existing studies and applications of product formulas, this order is fixed to be brickwall, while the actual effect of ordering remains mostly unexplored.

\begin{figure*}[t!]
  \centering
  \includegraphics[width=0.98\textwidth]{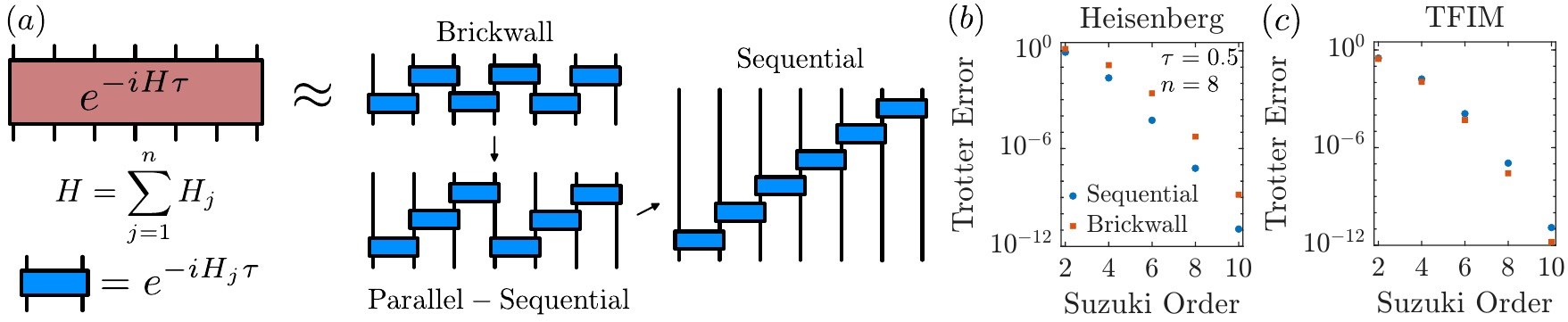}
  \caption{\textbf{Term ordering and Trotter error.} (a) Illustration of different term orderings within a first-order product formula. The same local evolutions \(e^{-iH_j\tau}\) can be arranged into different circuit geometries. Shown are representative examples, including brickwall and sequential circuits, together with an intermediate circuit from the parallel-sequential family~\cite{wei2025state}. (b,c) Trotter errors for the Heisenberg and the transverse-field Ising model,
  respectively, with \(n=8\) and \(\tau=0.5\). The relative accuracy of brickwall and sequential orderings depends on both the Hamiltonian and the Suzuki order.}
  \label{fig_illu}
\end{figure*}

Brickwall and sequential circuits differ structurally in how their gate orderings allow correlations to spread. A shallow nearest-neighbor brickwall circuit has a strict finite-depth light cone, whereas a sequential circuit can propagate correlations across the entire chain in a single sweep. This difference in correlation propagation is reflected, for example, in state preparation: sequential circuits can generate tensor-network states, including topologically ordered states~\cite{Schon2005,Schon2007,zypp,Banuls2008,chen2024sequential,yj_circ}, whose preparation with brickwall circuits requires at least linear depth~\cite{Bravyi2006}. This contrast provides intuition for investigating whether gate ordering might similarly affect the accuracy of product-formula simulation.

For Hamiltonian simulation, \cref{fig_illu}(b,c) gives a direct numerical illustration of the importance of this choice.
For the $8$-term ($9$-qubit) 1D Heisenberg model with step size $0.5$ simulated by a Suzuki formula, the sequential circuit has a smaller Trotter error than the brickwall circuit, and its advantage grows with the formula order, reaching about two orders of magnitude at order 10. For the transverse-field Ising model (TFIM), the trend is reversed: brickwall is more accurate, and its advantage also grows with the formula order. These examples show that ordering can have considerable Hamiltonian-dependent and order-dependent effects on the Trotter error. Also, it is unclear 
whether there might be a common optimal ordering for most Hamiltonians.

These observations raise two related questions that we address in this work:
\begin{description}[style=multiline, leftmargin=2.2cm]
\item[Question 1]
\emph{How can we efficiently estimate or bound the Trotter error of a given local Hamiltonian and product formula as accurately as possible, taking ordering into account?}
\item[Question 2]
\emph{How can we efficiently find the optimal ordering?
}
\end{description}

Previous work has studied certain limited types of ordering optimization:  Ref.~\cite{childs2019faster} conducted leading-order error analysis and investigated the potential benefits of randomized orderings in low-order product formulas;
Ref.~\cite{childs2019nearly} 
established ordering robustness of the asymptotic first-order error scaling for local Hamiltonians;
and Ref.~\cite{tranter2019ordering} studied term ordering for electronic-structure Hamiltonians using heuristic strategies, providing empirical evidence that ordering matters in certain settings.
Ref.~\cite{lane2025ordering} also considered ``ordering", but in a different sense: they studied the performance of alternative recursive decompositions of the Hamiltonian, which affects the formula structure and gate count.
However, we are not aware of systematic studies of ordering effects in Trotterization or the optimization of term ordering for a fixed higher-order product formula.

Existing product-formula analyses focus on the asymptotic Trotter error scaling, so they generally overestimate the higher-order Trotter error by orders of magnitude and pay little attention to ordering effects. For example, the commutator-scaling theory of Ref.~\cite{childs2021theory} gives
bounds of the schematic form
\begin{equation}
\begin{aligned}
 &\norm{e^{-i H\tau}-S^{(p)}(\tau)}\\
 &=\OC\left(\tau^{p+1}\sum_{j_1,\ldots,j_{p+1}}
 \norm{[H_{j_{p+1}},\ldots,[H_{j_2},H_{j_1}]\dots]}\right),
\end{aligned}
\label{eq:intro_standard_bound}
\end{equation}
where $S^{(p)}(\tau)$ is a $p$th-order product formula and $\tau\geq0$ denotes the step size.
The sum captures locality and commutation, but is unchanged by permuting the Hamiltonian terms. Moreover, bounding individual nested commutators separately can discard substantial cancellation. Also, \cref{eq:intro_standard_bound} only specifies the error scaling, while a quantitative numerical bound requires its prefactor.

More involved Trotter error analysis for given Hamiltonians and product formulas has a long history in the operator-splitting literature~\cite{huyghebaert1990product,jahnke2000error,thalhammer2008high,descombes2010exact,auzinger2013defect}.
In particular, exact local-error representations provide a starting point for bounds that do not rely on truncating a Taylor or Baker--Campbell--Hausdorff (BCH) series.
Ref.~\cite{childs2019nearly} used such a representation to give a bound for geometrically local systems.
Small-prefactor commutator bounds postpone some triangle inequalities, giving some of the tightest known bounds:
Proposition 16 of Ref.~\cite{childs2021theory} gives a slightly ordering-aware second-order bound; Appendix M of Ref.~\cite{childs2021theory} gives a tighter fourth-order bound; and Theorem 1 of \textcite{schubert2023trotter} applies to arbitrarily high orders and also gives different bounds for different orderings.
Nevertheless, a substantial gap between these bounds and numerical performance remains, especially at higher order. 
This gap is evident in the higher-order resource estimates of Refs.~\cite{childs2018toward,childs2019nearly,childs2021theory}; for example, the fourth-order bounds in Ref.~\cite{childs2021theory} predict about $5$--$7$ times as many Trotter steps as required numerically for a ten-qubit Heisenberg chain, corresponding to several orders of magnitude of overestimation in single-step Trotter error.
In a direct comparison of Trotter errors for a four-site Fermi--Hubbard model, Ref.~\cite{schubert2023trotter}
also found that their bounds overestimate the error by about one order of magnitude at second order and three orders of magnitude at fourth order.
Recent computational approaches improve on purely analytical estimates: Ref.~\cite{blunt2025monte} obtained tight commutator-norm bounds for second-order product formulas using projector Monte Carlo, while Ref.~\cite{maxwell2026practical} developed practical higher-order error estimates from compressed BCH expressions and importance sampling.
However, the latter estimates describe asymptotic error rather than a full finite-step operator-norm upper bound. In summary, obtaining practically tight, efficiently evaluable finite-step bounds for higher-order product formulas has been a longstanding open problem.

\subsection{Summary of results}
\subsubsection{Ordering-aware Trotter error bound}
Our first contribution is an ordering-aware
Trotter error bound for arbitrarily high-order product formulas, applicable to geometrically local Hamiltonians. For simplicity, we focus on 1D nearest-neighbor systems, but in principle our results can be generalized to higher constant dimensions, as we discuss in \cref{sec:discuss}.
Numerical tests show that our bound improves over the best previous bounds by several orders of magnitude for all examples tested for product formulas of order $p>2$ and only overestimates the actual Trotter error by roughly $40\%-50\%$, suggesting it is nearly optimal. 
The key idea is to organize error contributions by their spatial support rather than their Taylor order. An error term can extend over a large region only through a sequence of neighboring interactions or local gates, which suppresses contributions with large support.
It is therefore natural to compute the local contributions numerically, keeping all orders in time and their cancellations, and analytically bound the small nonlocal tails. 

The starting point of the ordering-aware
theory of Trotter error is a bound on the leading-order error.
For a $p$th-order formula applied to a nearest-neighbor Hamiltonian, this error consists of nested commutators of $p+1$ connected Hamiltonian terms.
We can therefore collect the leading-order contributions into local error blocks, each involving at most $p+1$ consecutive terms.
Within each block, we perform exact computation of the error contributions as matrices before taking a norm, keeping cancellations that would be lost by bounding nested commutators separately.
Summing these block norms gives an ordering-aware
leading-order bound.
This construction makes the spatial structure of the error explicit and suggests a way to go beyond leading order.

To obtain a bound on the full finite-time error, we use a similar idea of local evaluation. The strategy is to add the Hamiltonian terms one at a time and compare, at each addition, the changes in the exact and product-formula evolutions induced by the newly added term.
We refer to such a change as an \emph{increment}. Each increment is evaluated on a finite spatial window at the actual step size, so the calculation is efficient and includes all orders in time and the cancellations within that window.
Only the influence of the terms outside the window is bounded analytically, through local contributions at the boundary. 
The approximation is therefore in the spatial extent of the calculation, not in the Taylor order of the local error.
We further tighten the bound in two complementary ways: by grouping neighboring increments before taking a norm and by enlarging the exact-evaluation window. The former captures cancellation that would otherwise be lost when increments are bounded separately, while the latter moves more of the dynamics from the nonlocal correction into the exact local calculation. Varying the amount of grouping and the window size yields a family of efficiently computable bounds that trade classical computational cost for tighter error estimates.

We validate the effectiveness of our new Trotter error bounds by extensive numerical tests. For $500$ random translation-invariant $10$-qubit Hamiltonians, with normalized local terms and step size $0.1$, our fourth-order bound only overestimates the error by roughly $40\%$ on average for both sequential and brickwall circuits and correctly predicts which circuit performs better with $495/500$ success rate.
In comparison, for sequential circuits, Theorem 1 of Ref.~\cite{schubert2023trotter} and Appendix M of Ref.~\cite{childs2021theory} overestimate the error by three orders of magnitude and are unable to distinguish the two permutations.
At sixth order, we consider several larger systems of up to $256$ qubits and demonstrate that our Trotter error bounds are
more than six orders of magnitude tighter than previous works, leading to nearly optimal Trotter error bounds.

\subsubsection{Ordering optimization}
Our second contribution is an efficient graph-theoretic algorithm for optimizing the ordering on a translation-invariant chain. 
We also impose circuit-depth restrictions, which allow us to optimize the accuracy while maintaining a prescribed amount of parallelism.
We numerically test our algorithms and find that sequential or parallel-sequential circuits are generally optimal for random translation-invariant Hamiltonians.

For the ordering problem, an accurate bound alone does not make a direct search over permutations feasible, because a Hamiltonian with $n$ terms has $n!$ permutations.
On a nearest-neighbor chain, however, only the relative ordering of neighboring terms is needed to specify a product formula up to exchanges of commuting gates.
We show that $n!$ permutations are reduced to 
$2^{n-1}$ equivalence classes, indexed by a binary string $c=(c_1,\ldots,c_{n-1})$.
We then use the local structure of the leading-order error estimate introduced above to further reduce the search cost to $\poly(n)$: the leading-order bound for translation-invariant Hamiltonians is a sum of functions $f_{H,p}$ depending on $p$-bit substrings of $c$, 
\begin{equation}
 \text{leading-order error} \leq
 \tau^{p+1}\sum_j f_{H,p}(c_j,\ldots,c_{j+p-1}),
 \label{eq:intro_local_bound}
\end{equation}
and each $f_{H,p}$ is obtained from a calculation of cost $2^{\OC(p)}$.
Now, the bound corresponding to a long ordering can be built from a fixed set of local patterns.
We find a weighted \emph{de Bruijn graph} representation for the leading-order: local substrings are edges, and longer strings are walks.
We then prove that finding the permutation minimizing the leading-order error bound for asymptotically large systems reduces to finding a minimum-mean-cycle in the graph, which is efficiently solvable by Ref.~\cite{karp1978characterization} also in time $2^{\OC(p)}$, independent of $n$.
This representation is versatile: we can impose a depth restriction by excluding patterns with long sequential runs.

Our numerical tests show that this method selects a sequential ordering for almost all Hamiltonians, 
and that the selected circuits have errors close to the optimum over all permutations on the small chains that we can exhaustively simulate.
Under a depth restriction, the method predominantly selects wave-like circuits from the parallel-sequential family discussed above. These circuits preserve short sequential sweeps while allowing different parts of the chain to evolve in parallel.

The graph-based approach also extends to optimization of the full finite-time Trotter error bound. Each edge then represents an ordering of the terms in a local window and at its boundary. The weight represents the corresponding local error contribution plus the nonlocal correction. The graph for the full bound contains more vertices than the leading-order case. Nevertheless, in either case, the graph size is set by the range of the local calculation rather than by the total chain length.

The remainder of the paper is organized as follows. We introduce the necessary preliminaries in \cref{sec:preliminaries}, derive the local leading-order estimate in \cref{sec:leading}, and give the full Trotter error bound in \cref{sec:geometry}. We discuss a graph-theoretic representation of error estimates and an efficient algorithm to find the optimal ordering in \cref{sec:optimization}. In \cref{sec:numerics}, we present comprehensive numerical results to test our leading-order error estimates, geometry-aware error bounds, and graph theoretic permutation optimization.
Finally, in \cref{sec:nondegenerate}, we provide a rigorous proof for the advantage of sequential ordering in the nondegenerate part of the leading-order error.

\section{Preliminaries}
\label{sec:preliminaries}
\subsection{Permutations and local ordering patterns}
Let $S_n$ denote the symmetric group of $n$ items. For every $\pi \in S_n$, we denote
$\pi(j) \coloneqq \pi(1,2,\dots ,n)_j$.
\begin{definition}[Consecutive-label ordering pattern]
\label{def:cons-label}
    For each $\pi \in S_{n}$, its \emph{consecutive-label ordering pattern} (or simply \emph{ordering pattern}) $c(\pi)\in \{0,1\}^{n-1}$ satisfies
\begin{equation}
    c(\pi)_k \coloneqq 
    \begin{cases}
1 & \text{if $k$ appears before  $k+1$  in  $\pi(1,2,\dots,n)$,}  \\
0 & \text{otherwise}.
\end{cases}
\end{equation}
\end{definition}
For instance, when $n=5$, if $\pi(1,2,3,4,5)=(5,1,2,4,3)$, then $c(\pi) = 1100$, because $1$ appears before $2$, $2$ appears before $3$, $3$ appears after $4$, and $4$ appears after $5$. See \cref{fig:permutations}(a) for an illustration of this example.

\begin{definition}[Relative order of $k,k+1,\dots,k+l$]
\label{def:relativeorder}
    Given $k,l\in\mathbb{Z}$ satisfying $1\leq k\leq k+l \leq n $, we let $\sigma_{\pi,k,l}$ denote the permutation of $k,k+1,\dots,k+l$ within $\pi\in S_n$. In other words, $j\in\{k,k+1,\dots,k+l\}$ appears before $j'\in \{k,k+1,\dots,k+l\}$ in $\pi(1,2,\dots,n)$ if and only if $j$ appears before $j'$ in $\sigma_{\pi,k,l}(k,k+1,\dots,k+l)$.
\end{definition}

\begin{lemma}[Consecutive-label ordering of $k,k+1,\dots,k+l$]
\label{lem:relative_order}
The consecutive-label ordering pattern of the relative ordering of $k,k+1,\dots,k+l$ within $\pi$ is the substring of $c(\pi)$ from $k$ to $k+l-1$, i.e.,
    \begin{equation}
        c(\sigma_{\pi,k,l}) = \left(c(\pi)_k, c(\pi)_{k+1}, \dots, c(\pi)_{k+l-1} \right).
    \end{equation}
\end{lemma}
See \cref{fig:permutations}(b) for an illustration of how to read out relative order from the consecutive-label order.

\begin{definition}[$m$-translation-invariant permutations]
\label{def:tipermutation}
$\pi\in S_n$ is $m$-translation-invariant if $c(\pi)_j = c(\pi)_{j+m}$ for all $j$ such that both indices lie in $\{1,\ldots,n-1\}$.
\end{definition}
See \cref{fig:permutations}(c) for an example of a $3$-TI permutation.

\begin{figure}[t!]
    \centering
    \includegraphics[width=\linewidth]{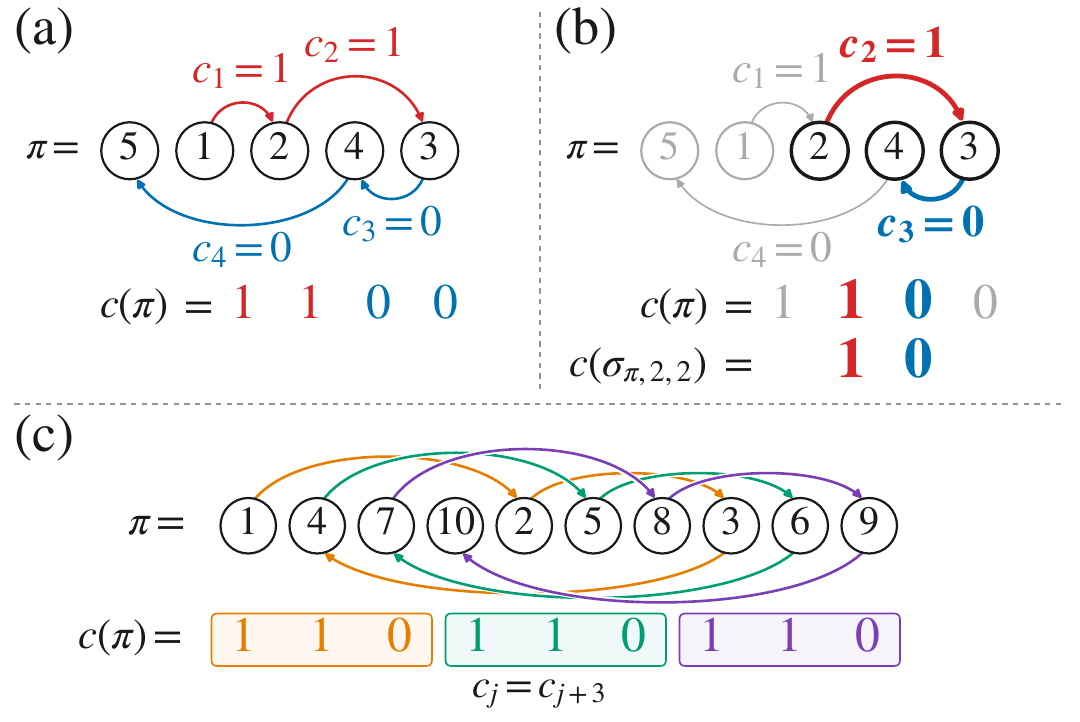}
    \caption{Illustration of the consecutive-label ordering pattern and its properties. Plot (a) shows how to compute $c(\pi)$, according to \cref{def:cons-label}; plot (b) shows how to read out the relative ordering between $2,3,4$ from $c(\pi)$, according to \cref{def:relativeorder} and \cref{lem:relative_order}; plot (c) shows a typical translation-invariant permutation with $m=3$, according to \cref{def:tipermutation}.}
    \label{fig:permutations}
\end{figure}

We also define the indicator function $\mathbf{1}_S$ for a set S as
\begin{equation}
\mathbf{1}_{S}
=
\begin{cases}
1, & x\in S,\\
0, & x\notin S.
\end{cases}    
\end{equation}

\subsection{Problem setup}
In this work, we first focus on 1D translation-invariant nearest-neighbor Hamiltonians with open boundary conditions. We study the worst-case performance, quantified using the spectral norm $\norm{\cdot}$.
We use $n$ to denote the number of local terms, so the chain has $n+1$ qubits.
\begin{definition}
    A Hamiltonian $H=\sum_j H_j$ is translation-invariant if there exists a 2-qubit Hermitian operator $h$ such that $H_j = I^{\otimes j-1}\otimes h \otimes I^{\otimes n-j}$.
\end{definition}
Note that $[H_j,H_{j+k}]=0$ for all $k\geq 2$ with $j+k\leq n$. 
For further simplicity, we also define the notations 
\begin{equation}
    H_{[j,k]}\coloneqq \sum_{i=j}^k H_i.
\end{equation}

We distinguish the \emph{term support} of a nested-commutator expression, meaning the set of labels of the $H_j$ appearing in it, from its \emph{qubit support}, on which the resulting operator acts nontrivially. The terms $H_j,\ldots,H_k$ act within qubits $j,\ldots,k+1$.

Throughout this paper, we use the forward product convention:
\begin{equation}
    \prod_{j=1}^k H_j \coloneqq H_1 H_2 \dots H_k.
\end{equation}
A general product formula with real coefficients containing $M$ local evolutions has the form
\begin{equation}
 S^{(p)}(\tau)=\prod_{\ell=1}^{M}e^{-i\delta_\ell H_{j_\ell}},
 \qquad \delta_\ell\coloneqq\alpha_\ell\tau.
 \label{eq:gate_times}
\end{equation}
Here $\tau$ is the physical step size, $\alpha_\ell$ is a fixed real coefficient, and $\delta_\ell$ is the signed duration of an individual gate. In particular, formulas can have negative $\delta_\ell$ (indeed, this is necessary for formulas of order greater than $2$ \cite{suzuki1991general}).
Note that $j_\ell$ and $\delta_\ell$ are not arbitrary assignments. Instead, they are a well-defined assignments for any certain $p$th-order product formula construction, such as the Suzuki-Trotter formula~\cite{suzuki1990fractal}.

\begin{definition}[Suzuki-Trotter formulas]
\label{def:suzuki}
    For $\tau>0$, 
    the second-order 
    Suzuki-Trotter product formula is
\begin{equation}
\label{eq:2nd_suzuki_0}
    S^{(2)}(\tau) = \prod_{j=1}^n e^{-iH_{j} \tau/2} \prod_{j=n}^1 e^{-iH_{j} \tau/2}.
\end{equation}
For $p\geq 4$ where $p$ is even, the $p$th-order Suzuki-Trotter product formula is defined recursively as
\begin{equation}
\label{eq:recursive_suzuki_0}
\begin{aligned}
    S^{(p)} (\tau) &= S^{(p-2)}(\mu_p\tau)   S^{(p-2)}(\mu_p\tau) \\
   & S^{(p-2)}  ((1-4\mu_p)\tau) S^{(p-2)}(\mu_p\tau) S^{(p-2)}(\mu_p\tau) 
\end{aligned}
\end{equation}
where
$\mu_p = (4 - 4^{1/(p-1)})^{-1}$.
\end{definition}

One can obtain the $\pi$-permuted Suzuki formula by permuting the gates in a second-order formula and to generating high-order versions recursively.
\begin{definition}[$\pi$-permuted Suzuki-Trotter formulas]
\label{def:permuted_pf}
For $\tau>0$ and $\pi \in S_n$, the permuted version of the $2$nd-order Suzuki formula $S^{(p)}$ is
\begin{equation}
\label{eq:2nd_suzuki}
    S^{(2)}(\tau) = \prod_{j=1}^n e^{-iH_{\pi(j)} \tau/2} \prod_{j=n}^1 e^{-iH_{\pi(j)} \tau/2}.
\end{equation}
For $p\geq 4$ where $p$ is even, the $p$th-order Suzuki-Trotter product formula is defined recursively as
\begin{equation}
\label{eq:recursive_suzuki}
\begin{aligned}
    S^{(p)} (\tau) &= S^{(p-2)}(\mu_p\tau)   S^{(p-2)}(\mu_p\tau) \\
   & S^{(p-2)}  ((1-4\mu_p)\tau) S^{(p-2)}(\mu_p\tau) S^{(p-2)}(\mu_p\tau) 
\end{aligned}
\end{equation}
where
$\mu_p = (4 - 4^{1/(p-1)})^{-1}$.
\end{definition}

Note that the permuted version can be defined generally for any product formulas $S^{(p)}(\tau)$ as
\begin{equation}
 S_\pi^{(p)}(\tau)=\prod_{\ell=1}^{M}e^{-i\delta_\ell H_{\pi(j_\ell)}}.
 \label{eq:permuted_PF_general}
\end{equation}
We define the additive Trotter error as
\begin{equation}
\label{eq:additive_error}
    \mathcal{A}_{p,\pi} \coloneqq e^{-i H \tau} - S_\pi^{(p)}(\tau).
\end{equation}
Since different permutations $\pi$ generally give different product formulas $S_\pi^{(p)}$, they also yield different additive Trotter errors.

The ordering-aware Trotter error bound developed in this work is applicable to any product formula. However, we only investigate the ordering optimization algorithm for Suzuki product formulas. In the remainder of this section, we introduce technical tools and results only applicable to Suzuki formulas.

Since non-neighboring $H_j$ terms commute in nearest-neighbor chains, many permutations are equivalent in terms of their corresponding Suzuki formulas.
For instance, if $|\pi(j)-\pi(j+1)|>1$, swapping $\pi(j)$ and $\pi(j+1)$ yields a $\pi'$ satisfying $S_{\pi'}^{(p)}=S_\pi^{(p)}$ for all Suzuki-Trotter formulas. This leads to \cref{lem:consec_label}.
\begin{lemma}
\label{lem:consec_label}
    If $c(\pi) = c(\pi')$, then $\prod_{j=1}^{n} \exp(iH_{\pi(j)}\delta_{\pi(j)}) = \prod_{j=1}^{n} \exp(iH_{\pi'(j)} \delta_{\pi(j)})$, 
    and $S_\pi^{(p)}(\tau) = S_{\pi'}^{(p)}(\tau)$ for Suzuki-Trotter formula $S_\pi^{(p)}(\tau)$. 
\end{lemma}
This lemma reduces $n!$ permuted Suzuki formulas to at most $2^{n-1}$ equivalence classes because $c$ is a $(n-1)$-bit string; additional commutation relations can make different patterns equivalent.
Henceforth, we use the consecutive-label ordering pattern $c(\pi)$ rather than $\pi$ to label the permuted Suzuki product formulas. We let $S_c^{(p)}(\tau)$ denote the $p$th-order Suzuki-Trotter formula corresponding to $c\in \{0,1\}^{n-1}$.
\begin{definition}
For Suzuki formulas, $S_c^{(p)}(\tau) = S_\pi^{(p)}(\tau)$ for all $\pi\in S_n$ satisfying $c(\pi)=c$.
\end{definition}

For the same reason, we use $\mathcal{A}_{p,c}$ instead of $\mathcal{A}_{p,\pi}$ to denote the corresponding Trotter error of a Suzuki formula.

\section{Ordering-aware leading-order error bound}
\label{sec:leading}
In this section, we develop an efficient strategy to estimate the leading-order Trotter error 
as a starting point towards full Trotter error bounds for arbitrary higher-order product formulas [cf.~\cref{sec:geometry}]. For simplicity and consistency with \cref{sec:optimization} which discusses ordering optimization, we focus on Suzuki-Trotter product formulas in this section.

We first show that, for $p$th-order Suzuki formulas, the leading-order contribution to $\mathcal{A}_{p,c}$, denoted by $\tilde{\mathcal{A}}_{p,c}$, is the sum of nested commutators of $p+1$ $H_j$ terms.
\begin{lemma}
\label{lem:everythingisNC}
    $\tilde{\mathcal{A}}_{p,c}$ is a linear combination of nested commutators of $p+1$ Hamiltonian terms.
\end{lemma}
\begin{proof}
    The Baker–Campbell–Hausdorff formula says that for two anti-Hermitian operators $A,B$, there exists an anti-Hermitian operator $C$ such that $e^{A}e^{B} = e^{C}$ with 
    \begin{equation}
        C=A+B+\frac{1}{2}[A,B]+\sum_{k=3}^\infty L_k(A,B), 
    \end{equation}
    where $L_k(A,B)$ is the $k$th-order homogeneous Lie polynomial, containing $k$th-order nested commutators.
    Recall from \cref{eq:gate_times} that $S_c^{(p)}$ is a product of gates $e^{-iH_{j_\ell}\delta_\ell}$, where $\delta_\ell=\alpha_\ell\tau$ is the signed gate duration. 
    Using the BCH formula, we can write $S_c^{(p)} =e^{-iJ}$
    with $J = \tau H + \tau^2 O_2 + \tau^3 O_3 + \cdots$, where $O_k$ is a linear combination of nested commutators of $k$ terms.
    Since $S_c^{(p)}$ is a $p$th-order product formula, $e^{-iJ} - e^{-iH\tau}$ does not contain any term of degree at most $p$ in $\tau$. Taylor expanding $e^{-iJ}$ shows that the leading-order error in $J$ must be a linear combination of nested commutators of $p+1$ terms. Therefore $J = \tau H + \tau^{p+1} O_{p+1} + \OC(\tau^{p+2})$, so
    \begin{align}
        \mathcal{A}_{p,c} &=e^{-iH\tau}-e^{-iJ} \nonumber\\
        &=  e^{-iH\tau} - \sum_{j=0}^\infty \frac{\left[-i (\tau H +\tau^{p+1} O_{p+1} + \OC(\tau^{p+2}))\right]^j}{j!} \nonumber\\
        &= e^{-iH\tau} - \sum_{j=0}^\infty \frac{(-i \tau H)^j}{j!} +i\tau^{p+1} O_{p+1} + \OC(\tau^{p+2})\nonumber\\
        &= i\tau^{p+1}O_{p+1} + \OC(\tau^{p+2}).
    \end{align}
\end{proof}

Let us say that $H_j, H_k$ are \emph{connected} if $|j-k|=1$, and that a subset of terms is connected if for any two terms in the subset, there exists a sequence of terms in the subset connecting them.
For a 1D nearest-neighbor Hamiltonian, we can first divide all possible $(p+1)$-term nested commutators into the following three types and analyze their properties.
\begin{itemize}
    \item \emph{Disconnected}: $p+1$ terms are not connected. All such nested commutators vanish, because when evaluating the nested commutator from the inside to the outside, there must exist at least one step where the LHS and the RHS of the commutator have non-overlapping supports. Example: $[H_2,[H_7,[H_5,[H_3, H_4]]]]$ vanishes when $H_7$ is added.
    \item  \emph{Degenerate}: $p+1$ terms are connected, and at least one term appears more than once in the nested commutator. Example: $[H_2,[H_1,[H_3, [H_2,H_1]]]]$. 
    \item \emph{Nondegenerate}: $p+1$ terms are connected, and every term appears only once in the nested commutator. Example: $[H_5,[H_4,[H_1, [H_3, H_2]]]]$
\end{itemize}

Note that for general many-body Hamiltonians, simple commutator counting suggests that nondegenerate terms dominate the leading-order expansion, and previous work on randomized orderings~\cite{childs2019faster} exploited the cancellation of nondegenerate contributions under randomization.
However, in geometrically local systems,
the numbers of both degenerate and nondegenerate commutators are strongly restricted. Therefore, in this work, we take both types into consideration to construct an ordering-aware leading-order Trotter error estimate.

Using \cref{lem:everythingisNC} and the observation that all disconnected nested commutators vanish, we can express the leading-order Trotter error as a linear combination of connected nested commutators: we first list all possible term supports, and then take the sum of all nested commutators on each term support.
Define the set of local nested commutators of $p+1$ terms supported from $H_j$ to $H_{j+k}$ for $k\leq p$ as
\begin{equation}
\begin{aligned}
    \mathrm{NC}_{p+1} (j,j+k) &\coloneqq \{ [H_{j_1},\dots,[H_{j_{p}},H_{j_{p+1}}]]:\\
    &\qquad \min(j_1,\dots,j_{p+1})=j,\\
    &\qquad \max(j_1,\dots,j_{p+1})=j+k,\\
    &\qquad [H_{j_1},\dots,[H_{j_{p}},H_{j_{p+1}}]]\\
    &\qquad \qquad \qquad \text{is connected}\}. 
\end{aligned}
\end{equation}
The leading-order additive Trotter error has the form
\begin{equation}
    \tilde{\mathcal{A}}_{p,c} = \tau^{p+1} \sum_{j=1}^n  \sum_{k=1}^{\min(p,n-j)} \sum_{L\in \mathrm{NC}_{p+1} (j,j+k)} \gamma_{L,c} L, 
\end{equation}
where $\gamma_{L,c} \in \mathbb{C}$.
Note that $j$ in the above summation denotes the smallest label in the term support of the nested commutator. 
We can then define the summand corresponding to a specific $j$ as the \emph{error block} starting at $j$, denoted by $\tilde{\mathcal{A}}_{p,c,j}$, so that $\tilde{\mathcal{A}}_{p,c} = \tau^{p+1}\sum_j \tilde{\mathcal{A}}_{p,c,j} $ with 

\begin{equation}
    \tilde{\mathcal{A}}_{p,c,j} \coloneqq   \sum_{k=1}^{\min(p,n-j)} \sum_{L\in \mathrm{NC}_{p+1} (j,j+k)} \gamma_{L,c} L.
\end{equation}
Note that all blocks labeled by $j\leq n-p$ contain exactly $p+1$ terms. In the limit that $n\gg 2^p$, the remaining blocks contribute negligibly to $\tilde{\mathcal{A}}_{p,c}$. Therefore, we focus on $j\leq n-p$ in this discussion.

For a given order $p$, every error block $\tilde{\mathcal{A}}_{p,c,j}$ is determined by the ordering pattern $c$ and the starting point $j$. However, the error blocks of two different ordering patterns $c^{(1)}$ and $c^{(2)}$ might be the same, as shown in the following lemma.
\begin{lemma}
\label{lemma:substring_c}
The error blocks starting at $j$ for ordering patterns $c^{(1)},c^{(2)}$ are the same,
i.e.,
    $
        \tilde{\mathcal{A}}_{p,c^{(1)},j} = \tilde{\mathcal{A}}_{p,c^{(2)},j}
    $,
    if
    \begin{equation}
        c^{(1)}_j, c^{(1)}_{j+1}, \dots, c^{(1)}_{j+p-1} 
      = c^{(2)}_j, c^{(2)}_{j+1}, \dots, c^{(2)}_{j+p-1}.
    \end{equation}
\end{lemma}
\begin{proof}
    The error block starting at $j$ involves only Hamiltonian terms from $H_j$ to $H_{j+p}$. Recall that $c_l$ indicates the relative order between $H_l$ and $H_{l+1}$. If $l \notin \{j,j+1,\dots,j+p-1\}$, then the value of $l$ does not affect the relative order of terms from $H_j$ to $H_{j+p}$. Therefore, the error block starting at $j$ depends only on the substring of the ordering pattern from $j$ to $j+p-1$, according to \cref{lem:relative_order}.
\end{proof}
Next, we focus on translation-invariant Hamiltonians. In this case, we have the following observation.
\begin{lemma}
\label{lemma:ti}
    If $H$ is translation invariant, then $\tilde{\mathcal{A}}_{p,c,j}$ and $\tilde{\mathcal{A}}_{p,c,k}$ are the same matrix acting on different qubits if $\left( c_j, c_{j+1}, \dots, c_{j+p-1} \right) = \left( c_k, c_{k+1}, \dots, c_{k+p-1} \right)$.
\end{lemma}
\begin{proof}
    This follows simply from the fact that every error block depends on the substring of the ordering pattern, and that $H$ is translation invariant.
\end{proof}
Combining \cref{lemma:substring_c,lemma:ti}, we obtain the following theorem.
\begin{theorem}
\label{thm:fHp}
    If $H$ is translation-invariant, there exists a function $f_{H,p}(c') \colon \{0,1\}^p \to \mathbb{R}$ that takes as input the relative order of $H_j,\dots, H_{j+p}$, i.e., $c'=(c_j,c_{j+1},\dots,c_{j+p-1})$ for any $j\leq n-p$, and gives the spectral norm of the error block starting at $j$, i.e.,
    \begin{equation}
        f_{H,p}(c') = \| \tilde{\mathcal{A}}_{p,c,j}  \|.
    \end{equation}
\end{theorem}

Suppose $f_{H,p}$ is given a priori. We can then apply the triangle inequality to construct an upper bound on the leading-order error that has localized sensitivity to the ordering pattern $c$:
\begin{equation}
\label{eq:newbound}
\begin{aligned}
    \| \tilde{\mathcal{A}}_{p,c} \| &\leq \tau^{p+1} \sum_{j=1}^n \| \tilde{\mathcal{A}}_{p,c,j} \| \\
    & = \tau^{p+1} \left[\sum_{j=1}^{n-p} f_{H,p}(c_j,\dots,c_{j+p-1})+b_{H,p}(c)\right].
\end{aligned}
\end{equation}
Here $b_{H,p}(c)\coloneqq\sum_{j=n-p+1}^{n}\|\tilde{\mathcal{A}}_{p,c,j}\|$ contains the shortened end windows at the right boundary of the chain. (If $n\leq p$, all blocks are treated as end windows.) For simplicity, we limit our analysis to $n\gg 2^p$, so that $b_{H,p}$ is negligible.

To instantiate this bound for a given $H$, we need to find the function $f_{H,p}$. 
Note that $c' \in \{0,1\}^p$ and every full error block involves at most $p+1$ consecutive terms, acting within $p+2$ qubits. We let $S^{(p)}_{c'}(\tau)$ be the $p$th-order product formula of $H_{[1,p+1]}$ corresponding to permutation pattern $c'$.
A straightforward way of computing $f_{H,p}(c')$ is to simply compute the leading-order Taylor expansion of $S^{(p)}_{c'}(\tau) - \exp(-iH_{[1,p+1]} \tau)$ and remove all terms not containing $H_1$. This strategy uses $2^{\OC(p)}$ classical computational effort.

We also propose a fully numerical method to compute $f_{H,p}$, which also takes $2^{\OC(p)}$ computation. We use the following observation: the leading-order terms of $S^{(p)}_{c'}(\tau) - \exp(-iH_{[1,p+1]} \tau)$ without $H_{1}$ are the leading-order terms of $S'^{(p)}_{c'}(\tau) - \exp(-iH_{[2,p+1]} \tau)$, where $S'^{(p)}_{c'}(\tau)$ is the same product formula as $S^{(p)}_{c'}(\tau)$ with every gate $e^{-iH_1\delta_\ell}$ generated by $H_1$ replaced by identity. Therefore, $f_{H,p}(c')$ can be computed via
\begin{equation}
\label{eq:fhp}
\begin{aligned}
    f_{H,p}(c')=\lim_{\tau\rightarrow0^+}\frac{1}{\tau^{p+1}}
      \Big\|&\left(S^{(p)}_{c'}(\tau)-e^{-iH_{[1,p+1]}\tau}\right)\\
    \qquad-&\left(S'^{(p)}_{c'}(\tau)-e^{-iH_{[2,p+1]}\tau}\right)\Big\|.
\end{aligned}
\end{equation}    

This requires classically simulating the perfect evolution and product formula of a $(p+2)$-qubit system. These matrices have $\sim 2^{2p+4}$ entries, and the total cost of the computation is $2^{\OC(p)}$, which is polynomial for $p=\OC(\log n)$. Note that $p$ is typically treated as a constant. More generally, to keep the gate complexity polynomial in $n$, one should take $p=\OC(\log n)$, since implementing a $p$th-order Suzuki product formula requires $\poly(n)\cdot 2^{\mathcal \OC(p)}$ gates.

This numerical strategy can give an even more powerful Trotter error estimation if we remove the $\tau\rightarrow 0$ limit, i.e., replace $f_{H,p}$ by
\begin{equation}
\begin{aligned}
 f^{\mathrm{num}}_{H,p}(\tau;c')
 =\Big\|&\left(S^{(p)}_{c'}(\tau)-e^{-iH_{[1,p+1]}\tau}\right)\\
 -&\left(S'^{(p)}_{c'}(\tau)-e^{-iH_{[2,p+1]}\tau}\right)\Big\|,
\end{aligned}
\label{eq:finite_local_predictor}
\end{equation}
and define the \emph{local Trotter error estimate} as
\begin{equation}
    \sum_{j=1}^{n-p} f^{\mathrm{num}}_{H,p}(\tau;c_j,\dots,c_{j+p-1}) + b^\mathrm{local}_{H,p}(\tau,c),
\end{equation}
where $b^\mathrm{local}_{H,p}(c)$
is negligible in the large-$n$ limit.
The quantity $f^{\mathrm{num}}_{H,p}(\tau;c')$ measures how the local Trotter error matrix changes when $H_1$ is added to the other terms in the window. It includes all Taylor orders of this change at the chosen step size, with their cancellations, rather than only the first nonzero coefficient. Its leading term is $\tau^{p+1}f_{H,p}(c')$. Because the time dependence is already included, no additional factor of $\tau^{p+1}$ is needed. This local estimator still does not include contributions involving terms outside the window and is not by itself a bound on the full finite-time error, but is still more informative than the leading-order bound. We control the nonlocal contributions in the next section.

There is also a simple trick to further tighten \cref{eq:newbound}. Grouping neighboring error blocks before taking a norm, we have 
\begin{align}
\label{eq:leading_pair}
 \left\|\sum_{j=1}^n\tilde{\mathcal A}_{p,c,j}\right\|
 & \leq \sum_{r=1}^{n/2}
 \left\|\tilde{\mathcal A}_{p,c,2r-1}
             +\tilde{\mathcal A}_{p,c,2r}\right\|\\
             & \leq \sum_{j=1}^n \| \tilde{\mathcal{A}}_{p,c,j}\|.\nonumber
\end{align}

Note that the sum of a pair of blocks can be computed from $p+2$ consecutive terms once each $\tilde{\mathcal{A}}_{p,c,}$ is computed. 
We use a similar grouping strategy in the full Trotter bound introduced in the next section.

In \cref{sec:numerics}, we numerically showcase the effectiveness of this ordering-aware bound. In the examples shown there, our leading-order error estimate from \cref{eq:newbound} more reliably predicts whether sequential or brickwall ordering performs better 
and estimates the magnitude of the difference for second-, fourth-, and sixth-order Suzuki formulas more accurately than the existing small-prefactor bounds~\cite{childs2021theory,schubert2023trotter}.

Given the strong empirical agreement between our bound on the leading-order error and the actual product-formula error,
minimizing \cref{eq:newbound} or \cref{eq:leading_pair} appears to be a useful proxy for minimizing the actual error.

Before considering how to choose an optimal ordering, in the next section we extend the above leading-order analysis to construct a full Trotter error bound (i.e.\ beyond the leading-order bound of this section) for geometrically local systems, giving much tighter bounds on the full error.

\section{Ordering-aware Trotter error bound}
\label{sec:geometry}

\begin{figure*}[ht]
    \centering
    \includegraphics[width=\linewidth]{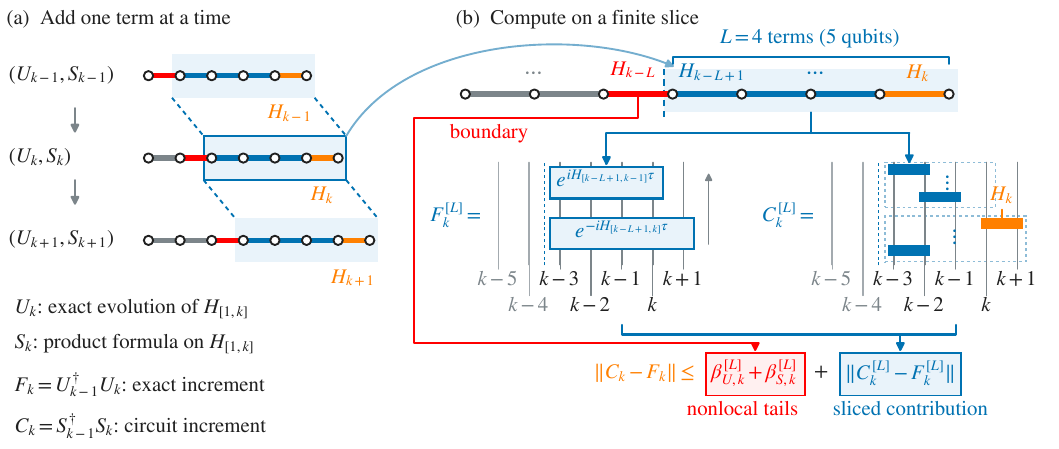}
    \caption{Construction of the ordering-aware Trotter error bound by adding one Hamiltonian term at a time. (a) Circles represent qubits and bonds represent nearest-neighbor Hamiltonian terms. The operators $U_k=e^{-iH_{[1,k]}\tau}$ and $S_k$ are the exact and product-formula evolutions of the first $k$ terms. The orange bond is the newly added term. Blue shading marks the last $L$ terms (the slice) used to approximate the increment.
    The red bond, corresponding to $H_{k-L}$ in the middle row, connects this slice to the rest of the chain. The dashed blue lines show the slice shifting by one term. The exact and circuit increments, $F_k=U_{k-1}^{\dagger}U_k$ and $C_k=S_{k-1}^{\dagger}S_k$, give $\|U_k-S_k\|\leq\|U_{k-1}-S_{k-1}\|+\|C_k-F_k\|$. (b) Each increment is approximated using the slice $H_{[k-L+1,k]}$, illustrated for $L=4$. The left diagram shows $F_k^{[L]}$ as two exact evolutions; the right shows $C_k^{[L]}$, obtained by deleting gates outside the slice from $S_{k-1}^{\dagger}S_k$, without changing the order or durations of the remaining gates. Operators are applied from bottom to top. The two leftmost wires are outside the slice and carry no gates; the orange gate acts on the support of $H_k$. The gate pattern is schematic, with ellipses indicating omitted gates. The blue contribution $\|C_k^{[L]}-F_k^{[L]}\|$ is computed numerically on the $L+1$ qubits of the slice. The red terms $\beta_{U,k}^{[L]}$ and $\beta_{S,k}^{[L]}$ bound $\|F_k-F_k^{[L]}\|$ and $\|C_k-C_k^{[L]}\|$, respectively. They are evaluated from boundary commutators on the slice together with $H_{k-L}$, involving at most $L+2$ qubits, as described in \cref{sec:Fbound,sec:Cbound}. Combining these contributions and iterating the inequality in (a) gives a bound on the full Trotter error. 
    }
    \label{fig:schematic_fullbound}
\end{figure*}
The guiding strategy of our full Trotter error bound is inspired by the leading-order error analysis: we decompose the Trotter error into a sum of local contributions, compute the local contributions via exact classical simulation of $\OC(p)$-qubit systems, and analytically bound the nonlocal tail that is not efficiently computable. 
We provide a visualization of this process in \cref{fig:schematic_fullbound}.

For notational consistency, we continue to denote the product formula and
its additive error by \(S_c^{(p)}\) and \(\mathcal{A}_{p,c}\). However, the
analysis in this section does not use any special property of Suzuki formulas
or of the ordering pattern \(c\). It applies to arbitrary product formulas
by replacing \(S_c^{(p)}\) and \(\mathcal{A}_{p,c}\) with \(S^{(p)}_\pi\) and
\(\mathcal{A}^{(p)}_\pi\). We continue using the \(c\)-labeled notation to be consistent with
\cref{sec:optimization} which later specializes to permuted Suzuki formulas, for which the consecutive-label ordering pattern is the natural optimization
variable.

\subsection{Local decomposition of Trotter error}
In a 1D chain, we consider the Trotter error $\mathcal{A}_{p,c}$ (see \cref{eq:additive_error}) by considering the incremental error when a new term is added to the Hamiltonian. In other words, we first consider the Trotter error of simulating $e^{-i H_1\tau}$ only (in this case, there is no Trotter error), and then the errors of simulating $e^{-i H_{[1,2]}\tau}$, $e^{-i H_{[1,3]}\tau}$, etc., until the error of simulating $e^{-i H_{[1,n-1]\tau}}$, and finally $e^{-i H \tau}$.

Let $k\in\{1,2,\dots,n\}$ be the index of a Hamiltonian term.
The exact evolution of the first $k$ terms is
\begin{equation}
    U_k(\tau) \coloneqq
    e^{-i H_{[1,k]} \tau}.
\end{equation}

Similarly to \cref{eq:fhp}, we define $S^{(p)}_{c,k}(\tau)$ by deleting from $S_c^{(p)}(\tau)$ every gate $e^{-iH_{j_\ell}\delta_\ell}$ whose term label satisfies $j_\ell>k$. In other words, following \cref{eq:gate_times}, if
$S_c^{(p)}(\tau)=\prod_{\ell=1}^{M}e^{-i\delta_\ell H_{j_\ell}}$, then correspondingly
\begin{equation}
    S^{(p)}_{c,k}(\tau) \coloneqq \prod_{\ell=1}^M e^{-i \mathbf{1}_{\{j_l\leq k\}}  \delta_l H_{j_l}}.
\end{equation}
To simplify the notation, we suppress the time, order, and permutation arguments, writing $S_k(\tau)=S^{(p)}_{c,k}(\tau)$. 
We define the accumulated additive error on simulating the first $k$ terms as
\begin{equation}
    \mathcal A_{c,k}^{(p)}(\tau) \coloneqq U_k(\tau) -S_k(\tau).
\end{equation}
To bound the norm of $\mathcal A_{c,k}^{(p)}(\tau)$, we define the \emph{evolution increment} $F_k(\tau)$ and \emph{product formula increment} $C_k(\tau)$:
\begin{equation}
\begin{aligned}
F_k(\tau)&\coloneqq U_{k-1}(\tau)^\dagger U_k(\tau),\\
C_k(\tau)&\coloneqq S_{k-1}(\tau)^\dagger S_{k}(\tau).
\end{aligned}
\end{equation}
We can now rewrite $U_k(\tau)$ and $S_{k}(\tau)$ in terms of these increments:
\begin{equation}
\begin{aligned}
 U_k(\tau)&=U_{k-1}(\tau)F_k(\tau),\\
 S_{k}(\tau)& =S_{k-1}(\tau)C_k(\tau).
\end{aligned}
\end{equation}

To further simply the notation, we also suppress the time argument, writing $S_k(\tau)=S_k$, $ U_k(\tau)=U_k$, $F_k(\tau)=F_k$, $ C_k(\tau)=C_k$, and $\mathcal A_{c,k}^{(p)}(\tau) = \mathcal A_{c,k}^{(p)}$. Now,
\begin{equation}
\begin{aligned}
\norm{\mathcal{A}_{c,k}^{(p)}}
&=\norm{U_{k-1}F_k-S_{k-1}C_k}\\
&\leq\norm{U_{k-1} - S_{k-1}}+\norm{C_k-F_k}\\
&=\norm{\mathcal{A}_{c,k-1}^{(p)}} +\norm{C_k-F_k}.
\end{aligned}
\label{eq:increment_recursion}
\end{equation}

Therefore, we call $C_k-F_k$ the \emph{error increment}.
Using \cref{eq:increment_recursion}, we can develop a Trotter error bound by simply summing bounds on the error increments:
\begin{equation}
\label{eq:thebound}
    \norm{\mathcal{A}_{c}^{(p)}} = \norm{\mathcal{A}_{c,n}^{(p)}} \leq \norm{\mathcal{A}_{c,L}^{(p)}} + \sum_{k=L+1}^n \norm{C_k - F_k}.
\end{equation}
Now, we analyze every error increment with $k\geq L+1$. When $k$ is large, both $C_k$ and $F_k$ are nonlocal unitaries supported on up to $k+1$ qubits. Therefore, brute-force computation of $\norm{C_k-F_k}$ is infeasible. However, intuitively, since $H$ is geometrically local, one can approximate $C_k-F_k$ by its localized contribution plus a nonlocal tail.

We introduce a \emph{slicing} strategy to separete local and nonlocal parts: we take a slice length $L$ 
and consider the contribution to $C_k-F_k$ supported on the last $L$ Hamiltonian terms (i.e., from $H_{k-L+1}$ to $H_k$). We aim to compute this contribution numerically, so $L=\OC(p)$ is a reasonable choice (in fact, we can simply choose $L=p+1$ because a leading-order contribution involves at most $p+1$ consecutive terms).
As shown in \cref{fig:schematic_fullbound}, we define the sliced product formula increment, $C_k^{[L]}$, by deleting from $C_k$ every gate $e^{-iH_{j_\ell}\delta_\ell}$ whose term label satisfies $j_\ell\leq k-L$, i.e., outside the window from $k-L+1$ to $k$. In other words, if
$C_k=\prod_{\ell=1}^{M}e^{-i\delta_\ell H_{j_\ell}}$ (it must have this form because $C_k$ is the product of two product formulas), then correspondingly
\begin{equation}
    C_k^{[L]} \coloneqq \prod_{\ell=1}^M e^{-i \mathbf{1}_{\{k-L<j_l\leq k\}}  \delta_l H_{j_l}}.
\end{equation}
Similarly, we define the sliced evolution increment
\begin{equation}
 F_k^{[L]}\coloneqq
 e^{iH_{[k-L+1,k-1]}\tau}e^{-iH_{[k-L+1,k]}\tau}.
\end{equation}
Note that if $k-L\leq 0$, there is no need to use the slicing startegy here, because $\|C_k-F_k\|$ can be directly computed. We illustrate the process of adding Hamiltonian terms into consideration in \cref{fig:schematic_fullbound}(a), and demonstrate how to compute sliced increments in \cref{fig:schematic_fullbound}(b).

Now, using triangle inequality, the error increment can be upper bounded by the sliced error increment plus the error in approximating $C_k,F_k$ using their sliced versions:
\begin{equation}
\norm{C_k-F_k}
\leq\norm{C_k-C_k^{[L]}}+\norm{C_k^{[L]}-F_k^{[L]}} + \norm{F_k-F_k^{[L]}}.
\end{equation}
Note that $C_k^{[L]}-F_k^{[L]}$ is the sliced error increment involving only $L$ Hamiltonian terms. With $L=p+1$, this matrix and its spectral norm can be computed with $2^{\OC(p)}$ complexity. Therefore, what remains to bound is the error in the slice approximation. We aim to find efficiently computable $\beta_{C,k}^{[L]}, \beta_{F,k}^{[L]}$ such that
\begin{equation}
\begin{aligned}
\beta_{C,k}^{[L]}&\geq \norm{C_k-C_k^{[L]}},\\
\beta_{F,k}^{[L]}&\geq \norm{F_k-F_k^{[L]}}.
\end{aligned}
\end{equation}
Thus,
\begin{equation}
\norm{C_k-F_k}\leq\beta_{C,k}^{[L]}+\beta_{F,k}^{[L]}
+\norm{C_k^{[L]}-F_k^{[L]}}.
\end{equation}
Next, we study how to find $\beta_{F,k}^{[L]}$ in \cref{sec:Fbound} and $\beta_{C,k}^{L}$ in \cref{sec:Cbound}.

\subsection{Bounding \texorpdfstring{$\norm{F_k-F_k^{[L]}}$}{the evolution increment error}}
\label{sec:Fbound}

We draw on the locality ideas underlying Lieb--Robinson bounds~\cite{lieb1972finite} and finite-region approximations to time evolution~\cite{barthel2012quasilocality} to propose the upper bound strategy.
In particular, we express $F_k - F^{[L]}_k$ as a time integral involving the commutator between the boundary term $H_{k-L}$ and time-evolved $H_k$ (i.e., $e^{isH_{[k-L+1,k-1]}}H_k e^{-isH_{[k-L+1,k-1]}}$). Rather than directly estimate this commutator through the general Lieb-Robinson bound, we exploit the finite window size $L$ and show that the commutator can be efficiently evaluated by a small-scale numerical simulation. 

As shown in the following lemma, the evolution increment is generated by the new term evolved under the preceding terms.

\begin{lemma}
\begin{equation}
\begin{aligned}
F_k(\tau)
&=\mathcal{T}\exp\left[
-i \int_0^\tau
e^{isH_{[1,k-1]}} H_k e^{-isH_{[1,k-1]}}\,ds
\right].
\end{aligned}
\end{equation}
\end{lemma}
\begin{proof}
A direct calculation gives
\begin{equation}
\begin{aligned}
\frac{dF_k(\tau)}{d\tau}
&=\frac{d e^{iH_{[1,k-1]}\tau}}{d\tau} U_k(\tau)
\\
&\quad
+U_{k-1}(\tau)^\dagger
\frac{d e^{-i(H_{[1,k-1]}+H_k)\tau}}{d\tau}
\\
&=-i e^{iH_{[1,k-1]}\tau}
H_k e^{-iH_{[1,k-1]}\tau} F_k(\tau).
\end{aligned}
\end{equation}
\end{proof}
Similarly, we can show that
\begin{equation}
\begin{aligned}
&\quad \, F_k^{[L]}(\tau)
\\
&=\mathcal{T}\exp\left[
-i \int_0^\tau
e^{isH_{[k-L+1,k-1]}}
H_k
e^{-isH_{[k-L+1,k-1]}}\,ds
\right].
\end{aligned}
\end{equation}
Therefore,
\begin{equation}
\begin{aligned}
&\norm{F_k(\tau)-F_k^{[L]}(\tau)}
\\
&\leq \int_0^\tau ds \, 
\Big\|
 e^{isH_{[1,k-1]}} H_k e^{-isH_{[1,k-1]}}\\
& \quad \quad \quad \quad \quad -e^{isH_{[k-L+1,k-1]}} H_k e^{-isH_{[k-L+1,k-1]}}
\Big\|.
\end{aligned}
\end{equation}
The first term in the integrand is nonlocal, and the second is local.
Our goal is to bound their difference. Set
\begin{equation}
\begin{aligned}
X_{[1,k-1]}(s)
&\coloneqq e^{isH_{[1,k-1]}} H_k e^{-isH_{[1,k-1]}},
\\
X_{[k-L+1,k-1]}(s)
&\coloneqq e^{isH_{[k-L+1,k-1]}}
H_k e^{-isH_{[k-L+1,k-1]}},
\\
Y_k(s)
&\coloneqq X_{[1,k-1]}(s)-X_{[k-L+1,k-1]}(s).
\end{aligned}
\end{equation}
Then
\begin{equation}
\begin{aligned}
\frac{dY_k}{ds}
&=i[H_{[1,k-1]},X_{[1,k-1]}]
\\
&\quad{}
-i[H_{[k-L+1,k-1]},X_{[k-L+1,k-1]}]
\\
&=i[H_{[1,k-1]},X_{[k-L+1,k-1]}+Y_k]
\\
&\quad{}
-i[H_{[k-L+1,k-1]},X_{[k-L+1,k-1]}]
\\
&=i[H_{[1,k-1]},Y_k]
\\
&\quad{}
+i[H_{[1,k-L]},X_{[k-L+1,k-1]}].
\end{aligned}
\label{eq:dyds}
\end{equation}
\cref{eq:dyds} is an inhomogeneous Heisenberg-type equation. Since $Y_k(0)=0$, Duhamel's formula gives
\begin{equation}
\begin{aligned}
Y_k(s)
&=i\int_0^s e^{i(s-u)H_{[1,k-1]}}
\\
&\quad{}\times
[H_{[1,k-L]},X_{[k-L+1,k-1]}(u)]
\\
&\quad{}\times
e^{-i(s-u)H_{[1,k-1]}}\,du.
\end{aligned}
\end{equation}
Now, observe that the only term in $H_{[1,k-L]}$ that does not commute with $X_{[k-L+1,k-1]}(u)$ is $H_{k-L}$. Consequently,
\begin{equation}
\begin{aligned}
&\quad \, \norm{Y_k(s)}
\\
&\leq \int_0^s
\norm{[H_{[1,k-L]},X_{[k-L+1,k-1]}(u)]}\,du
\\
&=\int_0^s
\left\|
\left[
H_{k-L},
e^{iuH_{[k-L+1,k-1]}}
H_k e^{-iuH_{[k-L+1,k-1]}}
\right]
\right\|\,du.
\end{aligned}
\end{equation}
It follows that
\begin{equation}
\begin{aligned}
& \quad \,  \norm{F_k(\tau)-F_k^{[L]}(\tau)}
\\
&\leq \int_0^\tau ds \, (\tau-s)
\\
& \quad \,  \left\| \left[
H_{k-L},
e^{isH_{[k-L+1,k-1]}}
H_k e^{-isH_{[k-L+1,k-1]}}
\right] \right\|.
\end{aligned}
\label{eq:betaF}
\end{equation}
Therefore, we can take $\beta_{F,k}^{[L]}$ as the right-hand side of the above equation.
The integrand of the right-hand side involves $L+1$ consecutive Hamiltonian terms, acting within $L+2$ qubits, and can thus be efficiently computed. When $k\leq L$, the boundary term $H_{k-L}$ does not exist, and $\norm*{F_k(\tau)-F_k^{[L]}(\tau)}=0$. Note that $\beta_{F,k}^{[L]}$ is the same for all product formulas, so it is the same for all orderings.
For translation-invariant Hamiltonians, $\beta_{F,k}^{[L]}$ is the same for all $k>L$. For non-translation-invariant Hamiltonians, $\beta_{F,k}^{[L]}$ must be computed for all $k$, contributing a factor of $n$ to the computational overhead.

\subsection{Bounding \texorpdfstring{$\norm{C_k-C_k^{[L]}}$}{the product formula increment error}}
\label{sec:Cbound}
We use a similar strategy to that in \cref{sec:Fbound} to find $\beta_{C,k}^{[L]}$ by computing commutators between sliced product formula increments and evolutions at the boundary of the slice. However, product formulas are different from continuous Hamiltonian evolutions, so we compare the complete and sliced increments gate by gate.

Recall that $C_k=S_{k-1}^\dagger S_k$. For a fixed $k>L$, let $M_k$ be the number of two-qubit evolutions in $S_k$, and write
\begin{equation}
\label{eq:Sk_as_gates}
\begin{aligned}
S_k&=G_1G_2\cdots G_{M_k},\\
G_\ell&=e^{-i\alpha_\ell\tau H_{j_\ell}},
\end{aligned}
\end{equation}
where $j_\ell\in\{1,\ldots,k\}$ and $\alpha_\ell$ is a real product formula coefficient, so that $\alpha_\ell\tau$ is the signed duration of the $\ell$th gate. The gates are indexed from left to right in this product. Since $S_{k-1}$ is obtained by deleting all $H_k$ evolutions from $S_k$, we can write
\begin{equation}
\label{eq:Sk-1_as_gates}
\begin{aligned}
S_{k-1}&=G'_1G'_2\cdots G'_{M_k},\\
G'_\ell&=
\begin{cases}
G_\ell, & j_\ell<k,\\
I, & j_\ell=k.
\end{cases}
\end{aligned}
\end{equation}

We now define the increment obtained from the first $\ell$ gates as
\begin{equation}
\label{eq:Ck_prefix}
W_\ell\coloneqq
(G'_1\cdots G'_\ell)^\dagger
(G_1\cdots G_\ell),
\qquad 0\leq\ell\leq M_k,
\end{equation}
where an empty product is the identity. Thus $W_0=I$ and $W_{M_k}=C_k$. If $j_\ell=k$, then $G_\ell$ appears only in $S_k$ and multiplies $W_{\ell-1}$ on the right. Otherwise, $G_\ell$ appears in both product formulas and conjugates $W_{\ell-1}$. Therefore,
\begin{equation}
\label{eq:Ck_prefix_recursion}
W_\ell=
\begin{cases}
W_{\ell-1}G_\ell, & j_\ell=k,\\
G_\ell^\dagger W_{\ell-1}G_\ell, & j_\ell<k.
\end{cases}
\end{equation}

Next, we derive a similar recursion for the sliced increment. Let $W_\ell^{[L]}$ be obtained from the expression in \cref{eq:Ck_prefix} by deleting all gates generated by $H_j$ with $j\leq k-L$. In particular, $W_0^{[L]}=I$ and $W_{M_k}^{[L]}=C_k^{[L]}$. The sliced recursion is
\begin{equation}
\label{eq:CkL_prefix_recursion}
W_\ell^{[L]}=
\begin{cases}
W_{\ell-1}^{[L]}G_\ell,
& j_\ell=k,\\
G_\ell^\dagger W_{\ell-1}^{[L]}G_\ell,
& k-L<j_\ell<k,\\
W_{\ell-1}^{[L]},
& j_\ell\leq k-L.
\end{cases}
\end{equation}

To bound $\norm*{C_k-C_k^{[L]}}$, we compare $W_\ell$ and $W_\ell^{[L]}$ from $\ell=1$ to $\ell=M_k$. If $j_\ell>k-L$, the two recursions apply the same right multiplication or unitary conjugation, so
\begin{equation}
\norm{W_\ell-W_\ell^{[L]}}
=
\norm{W_{\ell-1}-W_{\ell-1}^{[L]}}.
\end{equation}
Otherwise, $G_\ell$ is deleted from the sliced expression, giving
\begin{equation}
\begin{aligned}
W_\ell-W_\ell^{[L]}
={}&G_\ell^\dagger
(W_{\ell-1}-W_{\ell-1}^{[L]})G_\ell\\
&+G_\ell^\dagger W_{\ell-1}^{[L]}G_\ell
-W_{\ell-1}^{[L]}.
\end{aligned}
\end{equation}
Using the triangle inequality and unitary invariance of the spectral norm, we obtain
\begin{equation}
\label{eq:Ck_prefix_error}
\begin{aligned}
\norm{W_\ell-W_\ell^{[L]}}
\leq{}&
\norm{W_{\ell-1}-W_{\ell-1}^{[L]}}\\
&+\norm{\left[W_{\ell-1}^{[L]},G_\ell\right]}.
\end{aligned}
\end{equation}

Among all gates generated by $\{H_1,\ldots,H_{k-L}\}$, the commutator in \cref{eq:Ck_prefix_error} can be nonzero only for a gate generated by the boundary term $H_{k-L}$. Indeed, $W_{\ell-1}^{[L]}$ acts only on qubits $k-L+1,\ldots,k+1$, so gates generated by terms farther to the left commute with it. Thus each correction is computed in the slice together with the boundary term. The total number of qubits involved is $L+2$, so this commutator is efficiently computable.

Iterating over $\ell=1,\ldots,M_k$ and using $W_0=W_0^{[L]}=I$ gives
\begin{equation}
\label{eq:betaC}
\begin{aligned}
\norm{C_k-C_k^{[L]}}
&\leq
\sum_{\substack{1\leq\ell\leq M_k\\j_\ell=k-L}}
\norm{\left[W_{\ell-1}^{[L]},G_\ell\right]}.
\end{aligned}
\end{equation}
This right-hand side gives $\beta_{C,k}^{[L]}$, with $\beta_{C,k}^{[L]}=0$ when $k\leq L$. The contributions from all preceding $H_k$ evolutions are combined in $W_{\ell-1}^{[L]}$ before each commutator norm is taken, allowing cancellation that would be lost by bounding their contributions separately. Moreover, gates generated by terms farther to the left leave $W_\ell^{[L]}$ unchanged and contribute nothing to the bound, so the calculation only requires the gates generated by $H_{k-L},\ldots,H_k$.

Substituting the local calculation and the two nonlocal tails (\cref{eq:betaF,eq:betaC}) into \cref{eq:increment_recursion}, we obtain the efficiently computable full Trotter error bound
\begin{equation}
\begin{aligned}
\norm{\mathcal A_{c,n}^{(p)}} \leq \norm{\mathcal A_{c,L}^{(p)}} +\sum_{k=L+1}^n\big(&\|C_k^{[L]}-F_k^{[L]}\|\\
&+\beta_{F,k}^{[L]}+\beta_{C,k}^{[L]}\big).
\end{aligned}
\label{eq:q1}
\end{equation}

\subsection{Grouping neighboring increments: the \texorpdfstring{$q$}{q}-bound}
\label{sec:q2}
Similarly to the leading-order error estimate, the triangle inequality in \cref{eq:increment_recursion,eq:q1} loses cancellation between neighboring error increments. We can take this cancellation into consideration by adding $q\geq 2$ terms at a time when computing the local error increment.
For instance, we can define the $q=2$ evolution increment and $q=2$ product formula increment as
\begin{equation}
\label{eq:Fk2_Ck2}
 \begin{split}
 F_{k,2}&=U_{k-2}^\dagger U_k=F_{k-1}F_k,\\
 C_{k,2}&=S_{k-2}^\dagger S_k=C_{k-1}C_k.
 \end{split}
\end{equation}
For more general $q\geq 2$,
\begin{equation}
     \begin{split}
 F_{k,q}&=U_{k-q}^\dagger U_k=\prod_{j=k-q+1}^{k} F_j,\\
 C_{k,q}&=S_{k-q}^\dagger S_k=\prod_{j=k-q+1}^{k} C_j.
 \end{split}
\end{equation}
Then we can write down the $q$-counterpart of \cref{eq:increment_recursion}:
\begin{equation}
    \norm{\mathcal{A}_{c,k}^{(p)}} \leq \norm{\mathcal{A}_{c,k-q}^{(p)}} + \norm{C_{k,q}-F_{k,q}}.
\end{equation}
We can now define the sliced $q$-evolution increment for slice size $L\geq q$:
\begin{equation}
 F_{k,q}^{[L]}\coloneqq
 e^{iH_{[k-L+1,k-q]}\tau}e^{-iH_{[k-L+1,k]}\tau}.
\end{equation}
Similarly, the sliced product formula increment $C_{k,q}^{[L]}$ is obtained by excluding gates outside $\{H_{k-L+1},\ldots,H_k\}$ from $S_{k-q}^\dagger S_k$. Using the triangle inequality,
\begin{equation}
\begin{aligned}
    &\quad \,\, \norm{C_{k,q}-F_{k,q}}\\
    &\leq \norm{C_{k,q} -C_{k,q}^{[L]}} + \norm{C_{k,q}^{[L]}-F_{k,q}^{[L]}} + \norm{F_{k,q}-F_{k,q}^{[L]}}.
\end{aligned}
\end{equation}

We can still compute the sliced $q$-error increment $\|C_{k,q}^{[L]}-F_{k,q}^{[L]}\|$ via an $(L+1)$-qubit simulation. However, in this case, we should choose $L=p+q$ rather than $p+1$ to ensure that all error contributions supported on $p+1$ consecutive terms including $H_{k-q+1}$ are included in the computation.

The upper bound of $\norm*{F_{k,q}-F^{[L]}_{k,q}}$, denoted by $\beta^{[L]}_{F,k,q}$, is the right-hand side of \cref{eq:betaF} with $H_k$ replaced by
$H_{[k-q+1,k]}$ and $H_{[k-L+1,k-1]}$ replaced by $H_{[k-L+1,k-q]}$, i.e.,
\begin{equation}
\begin{aligned}
& \quad \,  \norm{F_{k,q}-F^{[L]}_{k,q}}
\\
&\leq \int_0^\tau ds \, (\tau-s)  \Big\| \big[
H_{k-L},\\
& \quad \quad \quad e^{isH_{[k-L+1,k-q]}}
H_{[k-q+1,k]} e^{-isH_{[k-L+1,k-q]}}
\big] \Big\|.
\end{aligned}
\label{eq:betaFq}
\end{equation}
The upper bound of $\norm*{C_{k,q}-C^{[L]}_{k,q}}$, denoted by $\beta^{[L]}_{C,k,q}$, can be obtained by the following generalization. Using the gate decomposition $S_k=G_1\cdots G_{M_k}$ in \cref{eq:Sk_as_gates}, we now obtain $S_{k-q}$ by deleting all gates generated by $H_{k-q+1},\ldots,H_k$. Thus we redefine $G'_\ell$ as
\begin{equation}
\begin{aligned}
G'_\ell&=
\begin{cases}
G_\ell, & j_\ell\leq k-q,\\
I, & j_\ell>k-q,
\end{cases}\\
S_{k-q}&=G'_1G'_2\cdots G'_{M_k}.
\end{aligned}
\end{equation}
Following \cref{eq:Ck_prefix}, define
\begin{equation}
\label{eq:Ckq_prefix}
W_{\ell,q}\coloneqq
(G'_1\cdots G'_\ell)^\dagger
(G_1\cdots G_\ell),
\qquad 0\leq\ell\leq M_k.
\end{equation}
Then $W_{0,q}=I$ and $W_{M_k,q}=C_{k,q}$. Gates generated by $H_{k-q+1},\ldots,H_k$ multiply the prefix increment on the right, while gates generated by $H_1,\ldots,H_{k-q}$ conjugate it.

Let $W_{\ell,q}^{[L]}$ be the sliced version of $W_{\ell,q}$, obtained by deleting all gates generated by $H_j$ with $j\leq k-L$. Its recursion is
\begin{equation}
\label{eq:CkqL_prefix_recursion}
W_{\ell,q}^{[L]}=
\begin{cases}
W_{\ell-1,q}^{[L]}G_\ell,
& k-q<j_\ell\leq k,\\
G_\ell^\dagger W_{\ell-1,q}^{[L]}G_\ell,
& k-L<j_\ell\leq k-q,\\
W_{\ell-1,q}^{[L]},
& j_\ell\leq k-L,
\end{cases}
\end{equation}
with $W_{0,q}^{[L]}=I$ and $W_{M_k,q}^{[L]}=C_{k,q}^{[L]}$. The same gate-by-gate comparison as in \cref{eq:Ck_prefix_error} shows that only gates generated by the boundary term $H_{k-L}$ contribute to the bound. Therefore, \cref{eq:betaC} generalizes to
\begin{equation}
\begin{aligned}
\norm{C_{k,q}-C_{k,q}^{[L]}}
&\leq
\sum_{\substack{1\leq\ell\leq M_k\\j_\ell=k-L}}
\norm{\left[W_{\ell-1,q}^{[L]},G_\ell\right]}.
\end{aligned}
\label{eq:betaCq}
\end{equation}
This right-hand side gives $\beta_{C,k,q}^{[L]}$. Each commutator acts on at most $L+2$ qubits and is efficiently computable. The contributions from all preceding evolutions generated by $H_{k-q+1},\ldots,H_k$ are combined in $W_{\ell-1,q}^{[L]}$ before each commutator norm is taken. As before, we set $\beta_{F,k,q}^{[L]}=\beta_{C,k,q}^{[L]}=0$ when $k\leq L$.

Finally, we choose the unique \(k_0\in\{L-q+1,\ldots,L\}\) such that $q$ divides $n-k_0$, and our $q$-Trotter error bound is
\begin{equation}
\begin{aligned}
 \norm{\mathcal A_{c,n}^{(p)}}&\leq
 \norm{\mathcal A_{c,k_0}^{(p)}}\\
 &\quad+\sum_{\substack{k=k_0+q,k_0+2q,\dots,n}}
 \big(\|C_{k,q}^{[L]}-F_{k,q}^{[L]}\|\\
 &\hspace{7em} +\beta_{F,k,q}^{[L]}+\beta_{C,k,q}^{[L]}\big).
\end{aligned}
\label{eq:q2}
\end{equation}

To see why this $q=2$ bound is potentially tighter than the $q=1$ bound of \cref{eq:q1}, we notice that, according to \cref{eq:Fk2_Ck2}, the sliced error increment in the $q=2$ case satisfies
\begin{equation}
\begin{aligned}
&\quad \,\, \norm{C^{[L]}_{k,2}-F^{[L]}_{k,2}} \\
&=\norm{C^{[L-1]}_{k-1}C^{[L]}_{k}
-F^{[L-1]}_{k-1}F^{[L]}_{k}}\\
 &=\norm{(C^{[L-1]}_{k-1}-F^{[L-1]}_{k-1})C^{[L]}_{k}
+F^{[L-1]}_{k-1}(C^{[L]}_{k}-F^{[L]}_{k})}\\
&\leq \norm{C^{[L-1]}_{k-1}-F^{[L-1]}_{k-1}} +\norm{C^{[L]}_{k}-F^{[L]}_{k}},
\end{aligned}
\label{eq:transported_pair}
\end{equation}
following the same logic as \cref{eq:leading_pair}. Here the first increment has slice size $L-1$ so that both factors have the same left boundary $k-L+1$, and all operators are embedded in the common $(L+1)$-qubit space. (The same logic also applies to larger $q$, which can potentially give tighter bounds at the cost of more expensive computations.)
However, although all our numerical tests show that the $q=2$ bound is tighter than the $q=1$ bound, this is not rigorously guaranteed because the bounds for different $q$ values have different slice sizes and nonlocal tails. In practice, one can simply take the smaller of the two complete valid bounds.

Note that using general $q$ with $L=p+q$ yields $2^{\OC(p+q)}$ classical computation cost. The resulting bound may allow for more cancellation in the exact calculation of the sliced Trotter error, but it comes with a computational cost that grows exponentially in $q$.

We present numerical validations of our $q=2$ bound (\cref{eq:q2}) in \cref{sec:numerics}. The $q=2$ bound is substantially tighter than existing bounds for all translation-invariant 1D examples tested. It only overestimates the true Trotter error by about 37--51\% in those tests, achieving orders of magnitude improvement over Refs.~\cite{childs2021theory,schubert2023trotter}. Furthermore, thanks to the exact (but efficient) simulation of sliced product formulas, this bound is ordering-aware: in particular, for almost all sampled translation-invariant Hamiltonians, it correctly predicts whether brickwall or sequential circuits yield smaller Trotter error. Therefore, in the next section we develop an efficient algorithmic framework to find the optimal gate ordering.

\subsection{Enlarging the slice: the \texorpdfstring{$(q,v)$}{(q,r)}-bound}
\label{sec:qr_bound}

The choice $L=p+q$ in \cref{sec:q2} can be further generalized to
\begin{equation}
    L=p+q+v,\qquad v=0,1,2,\ldots.
\end{equation}
We call \cref{eq:q2} with this slice size the \emph{$(q,v)$-Trotter error bound}. The sliced increments and the nonlocal tail bounds in \cref{eq:betaFq,eq:betaCq} are evaluated with this larger $L$, without changing their definitions. The case $v=0$ recovers the standard $q$-bound.

While increasing $q$ groups more neighboring increments, increasing $v$ extends the slice farther to the left without changing the newly added terms. The influence of the additional terms is then included in the joint calculation of $\norm*{C_{k,q}^{[L]}-F_{k,q}^{[L]}}$ rather than bounded separately in the two nonlocal tails. This allows for more cancellation and can improve the bound when the nonlocal tails are significant.

The local calculations involve at most $p+q+v+2$ qubits, giving a classical computation cost of $2^{\OC(p+q+v)}$ per window. Increasing $v$ does not guarantee a tighter complete bound, so in practice one can take the smallest bound among the values of $v$ evaluated.

The two parameters, $q$ and $v$, characterize a large family of Trotter error bounds. Since they are all valid upper bounds, one can always try multiple of them and use the numerically smallest one and the tightest bound available. We leave the optimal choice of $(q,v)$ under given classical computational resources as an open problem.

For simplicity, most of the numerical tests in \cref{sec:num_bounds} use the $q=2,v=0$ bound, which we henceforth refer to simply as the $q=2$ bound. We also test the $q=2,v=1$ bound for sixth-order product formulas simulating the Heisenberg model, and find that it further improves the estimated number of Trotter steps to reach a given target precision.

\section{Graph-theoretic algorithms for ordering optimization}
\label{sec:optimization}
In this section, we turn from Trotter error bounds to ordering optimization. In particular, we propose an efficient algorithm to find the consecutive-label ordering pattern $c$ (or equivalently, the permutation up to exchanges of commuting gates) that minimizes our Trotter error bounds. 

\subsection{De Bruijn graphs for the leading-order bound}
For simplicity, we focus on minimizing the leading-order error estimate (\cref{eq:newbound}) 
first and postpone the construction based on the full Trotter error bounds (\cref{eq:q1,eq:q2}) to \cref{sec:finite_pattern_cost}.

We first formalize the problem as follows. Given a system of $n$ terms and formula order $p\geq 2$ with $n\gg 2^p$, we aim to find $c\in\{0,1\}^{n-1}$ such that $\sum_{j=1}^{n-p} f_{H,p}(c_j,\dots,c_{j+p-1})$ (the bulk contribution to the leading-order bound in~\cref{eq:newbound}) is minimized.
We define the key graph-theoretic tool, the \emph{de Bruijn graph}, corresponding to the local leading-order error contribution $f_{H,p}(c')$ as follows.
\begin{definition}
    The de Bruijn graph corresponding to the leading-order Trotter error analysis of a $p$th-order product formula is a weighted graph whose vertices are $(p-1)$-bit strings, i.e.,
    \begin{equation}
     V = \{0,1\}^{p-1},   
    \end{equation}
    with directed edges
    \begin{equation}
    \begin{aligned}
    E =  &\{((v_1,v_2,\dots,v_{p-1}),(v_2,v_3,\dots,v_{p-1},x)):\\
    &\, \, x\in\{0,1\}, v_j\in\{0,1\}, 1\leq j\leq p-1 \}.
    \end{aligned}
    \end{equation}
    Each edge $s\rightarrow t$ 
    has weight (see \cref{thm:fHp})
    \begin{equation}
        w(s,t) = f_{H,p}(s_1,s_2,\dots,s_{p-1},t_{p-1}).
    \end{equation}
\end{definition}

\begin{figure}
    \centering
    \includegraphics[width=\linewidth]{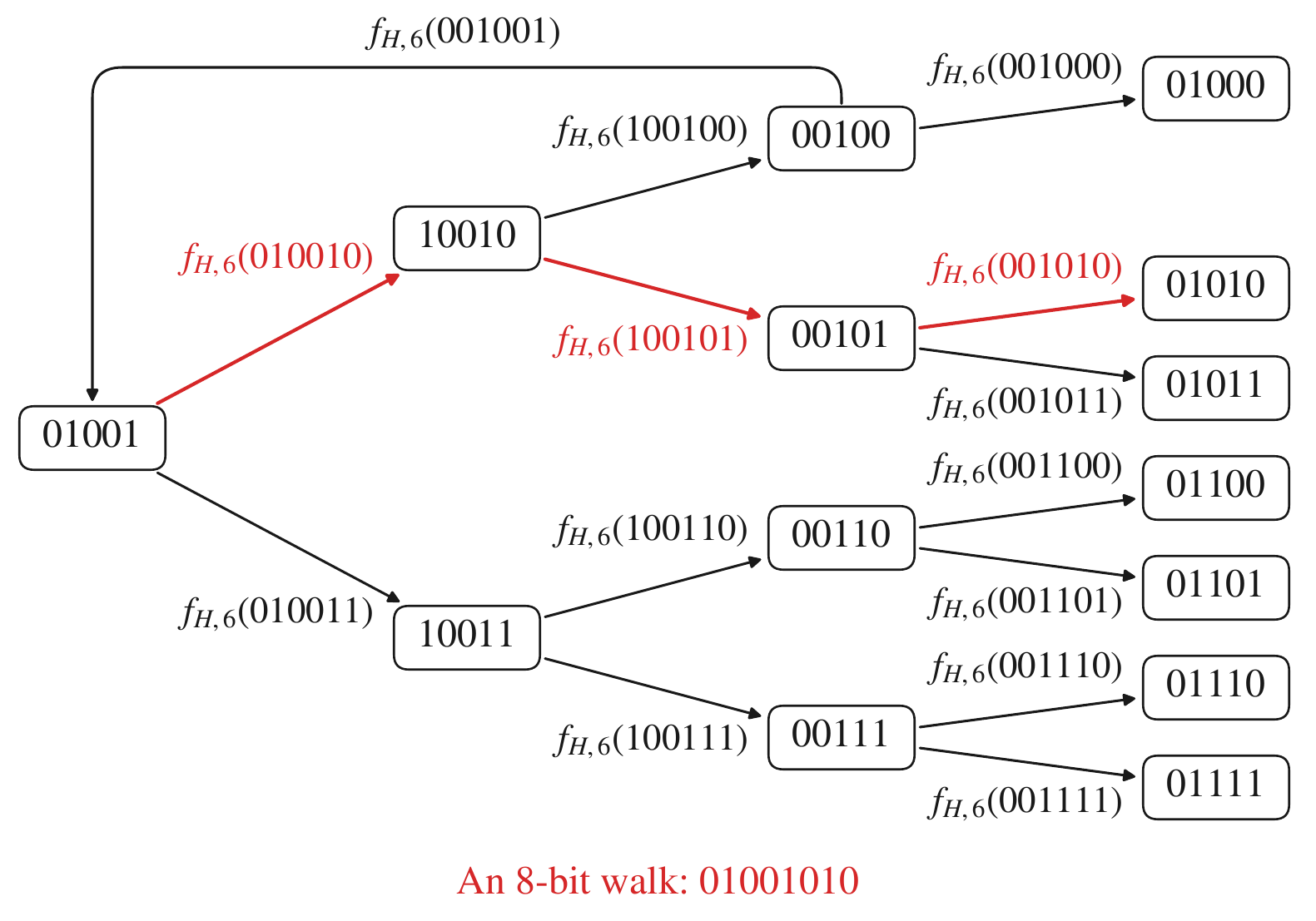}
    \caption{Part of the de Bruijn graph corresponding to $p=6$. Every vertex is a $(p-1)$-bit string, every edge corresponds to a $p$-bit string $c'$ with weight $f_{H,6}(c')$, and every longer bit string corresponds to a unique walk in the graph. The example presented in the graph is a walk of $8$ bits (i.e., a walk corresponding to an $8$-bit ordering pattern) with corresponding edges shown in red.}
    \label{fig:debruijn}
\end{figure}

\begin{figure}
    \centering
    \includegraphics[width=\linewidth]{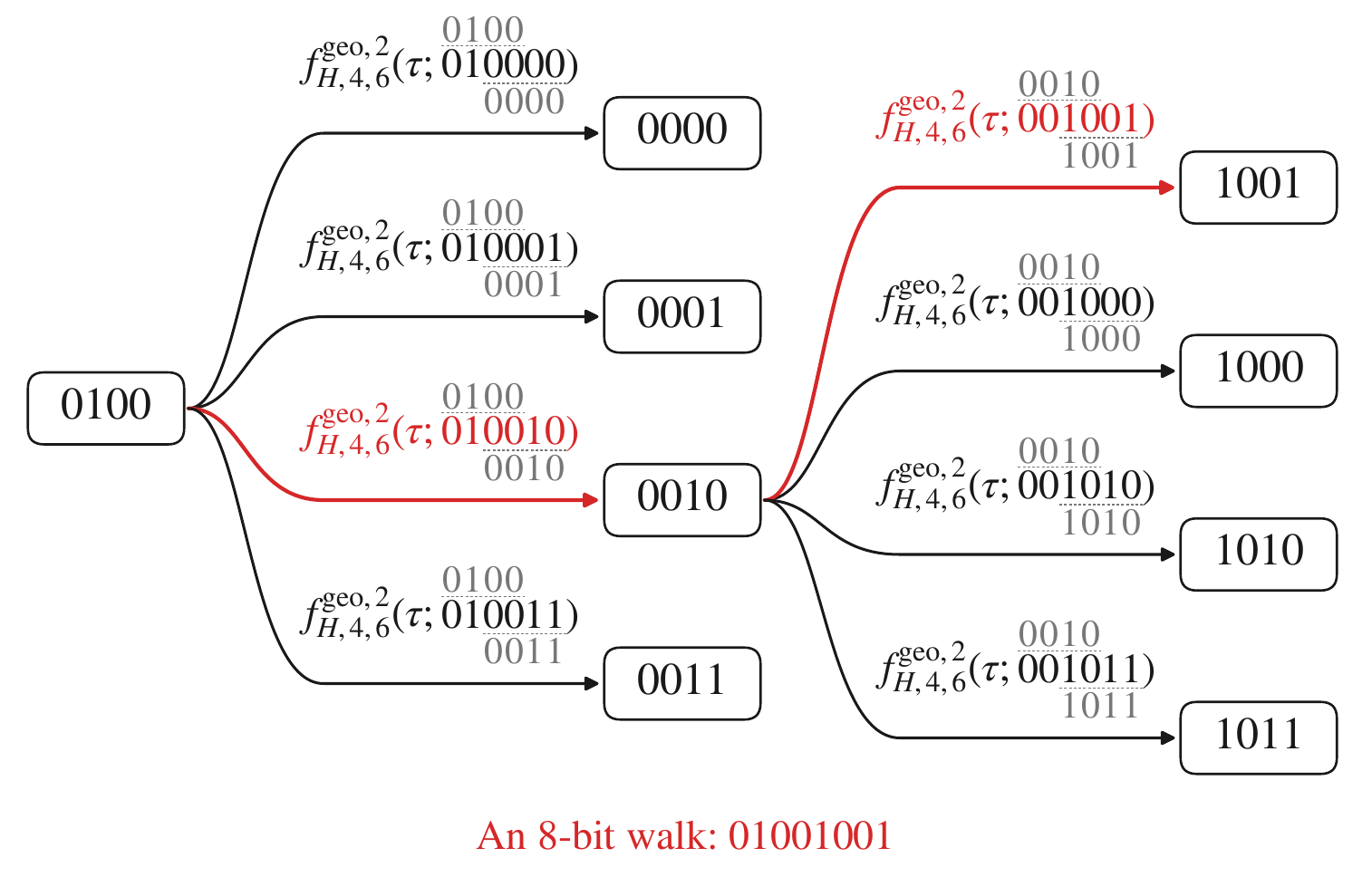}
    \caption{Part of the de Bruijn graph corresponding to $p=4$ based on the $q=2$ Trotter-error bound. Every vertex is a $p$-bit string, every edge is a $(p+q)$-bit string $c'$ with weight determined by $f_{H,p,L}^{\mathrm{geo},q}(\tau;c')$, and every longer bit string corresponds to a unique walk in the graph. The example presented in the graph is an $8$-bit walk
    with corresponding edges shown in red.}
    \label{fig:debruijnq}
\end{figure}

See \cref{fig:debruijn} for an illustration of a part of the de Bruijn graph for $p=6$.
Observe that the total number of length-$p$ substrings in $c$ is $n-p$, and every $c$ corresponds to a unique directed walk from $(c_1,\dots,c_{p-1})$ to $(c_{n-p+1},\dots,c_{n-1})$, containing $n-p+1$ vertices and $n-p$ edges. Conversely, every valid directed walk in the de Bruijn graph $(V,E)$ gives a valid ordering pattern $c$, corresponding to a class of equivalent term orderings. Therefore, the sum $\sum_{j=1}^{n-p}|f_{H,p}(c_j,\dots,c_{j+p-1})|$ is precisely the sum of edge weights along this walk, and the optimization problem becomes finding the minimum-weight walk with $n-p$ edges in $(V,E)$.

\subsection{De Bruijn graphs for the \texorpdfstring{$q$}{q}-Trotter error bounds}
\label{sec:finite_pattern_cost}
The same construction also applies to the full Trotter error bounds with general $q\geq 1$ and $v\geq 0$ for translation-invariant Hamiltonians. We now choose the graph weights to be the finite-window contributions in \cref{eq:q1,eq:q2}, rather than the leading-order contributions. For the permuted product formula at fixed $p,\tau$ and the $(q,v)$-Trotter error bound, using the setting that $L=p+q+v$, we can define a new weight function that takes $L$ bits as input:
\begin{equation}
\begin{aligned}
 &f^{\mathrm{geo},q}_{H,p,L}(\tau;c_{k-L},\ldots,c_{k-1})\\
 &\quad\coloneqq\norm{C_{k,q}^{[L]}-F_{k,q}^{[L]}}
            +\beta_{F,k,q}^{[L]}+\beta_{C,k,q}^{[L]}.
\end{aligned}
\label{eq:geo_local_cost}
\end{equation}
This weight still depends on $L$-bit substrings of $c$ for the following reasons. The numerical calculation of the sliced error increment, $\norm*{C_{k,q}^{[L]} - F_{k,q}^{[L]}}$, involves $L$ terms $H_{k-L+1},\ldots,H_k$, whose ordering is specified by $L-1$ bits. However, recall from \cref{eq:betaCq} that the computation of $\beta_{C,k,q}^{[L]}$ involves nontrivial commutators with $e^{i \delta_{j,l} H_{k-L}}$ gates in the product formula. Therefore, the relative order with $H_{k-L}$ should also be taken into account. This implies that we need to consider the relative order of $\{H_{k-L},\dots,H_{k}\}$ for every $k$, which involves $L+1$ terms in total and can be encoded in $L$ bits.

A key observation is that, for $q>1$, the weights in \cref{eq:q2} are added after every $q$ new terms.
This means that every vertex should be a string of $L-q = p+v$ bits, and each edge should correspond to a string of $L$ bits (denoted by $c'$), connecting the vertex of $(c'_1,\dots,c'_{p+v})$ to the vertex of $(c'_{q+1},\dots,c'_L)$. Note that if $q>p+v$, there can be multiple edges connecting two vertices.

Now, we can define the modified de Bruijn graph, the \emph{$q$-step de Bruijn graph}, for a generalized $(q,v)$-Trotter error bound.
\begin{definition}
    \label{def:debruijnq}
    The $q$-step de Bruijn graph corresponding to a $(q,v)$-Trotter error bound of a $p$th-order product formula is a weighted graph whose vertices are $(p+v)$-bit strings, i.e.,
    \begin{equation}
     V = \{0,1\}^{L-q} = \{0,1\}^{p+v},   
    \end{equation}
    with edges
    \begin{equation}
    \begin{aligned}
    E = \{0,1\}^{L} = \{0,1\}^{p+q+v},
    \end{aligned}
    \end{equation}
    where every edge $e\in E$ is directed from vertex $(e_1,\dots,e_{p+v})$ to vertex $(e_{q+1},\dots,e_L)$, and has weight
    \begin{equation}
        w(e) = f^{\mathrm{geo},q}_{H,p,L}(\tau;e_1,e_2,\dots,e_{L}).
    \end{equation}
\end{definition}
See \cref{fig:debruijnq} for the illustration of a part of the de Bruijn graph corresponding to the $q=2$ Trotter error bound (with $v=0$).
The permutation pattern $c$ minimizing the general $(q,v)$ error bound corresponds to a walk containing $ (n-k_0)/q$ edges (see \cref{eq:q2}). Neglecting the bounded initial contribution in the long-chain limit, optimizing the bulk $(q,v)$-bound corresponds to finding the minimum-weight walk with $(n-k_0)/q$ edges in the de Bruijn graph defined in \cref{def:debruijnq}.

\subsection{Optimal permutation via minimum-mean-weight cycle}
If we use brute-force search to find the optimal consecutive-label ordering pattern of length $n-1$, the computational cost scales exponentially with $n$. 
With the de Bruijn graph formulation, it can be found by dynamic programming in $\OC(n|E|)$ time,
via computing the minimum cost for each vertex at each walk length.
For $ 2^p,2^q \ll n$, we can further simplify the analysis and obtain an $n$-independent algorithm. For simplicity, we focus on optimizing the leading-order error bound for the rest of the section. All results also apply uniformly to arbitrary $q$-Trotter error bounds (\cref{eq:q2}).
\begin{lemma}
As $n$ increases, the minimum mean weight of a walk of length $n-p$ converges to the minimum cycle mean. Repeating a minimum-mean-weight cycle attains this value asymptotically.
\end{lemma}
\begin{proof}
    For a walk with $n-p$ edges, we list all edges in the walk as a sequence $(e_1, e_2,\dots,e_{n-p})$ and its vertices as $(v_0,v_1,\ldots,v_{n-p})$. Now, we can repeatedly run the following procedure: find a pair of vertices $v_j,v_k$ satisfying $v_j=v_k$ with $j<k$, and remove the closed subwalk $e_{j+1},e_{j+2},\ldots,e_k$. In the end, the remaining sequence contains no repeated vertex, and it has at most $|V|-1=2^{p-1}-1$ edges. Since $n\gg 2^p$, the total weight is dominated by the removed cycles rather than by the remaining edges, and this effect becomes more significant as $n$ increases.

    We call a cycle with vertices $v_0,v_1,\ldots,v_k$ a \emph{simple} cycle if $v_0=v_k$ and $v_0,\ldots,v_{k-1}$ are distinct. The length of a simple cycle is upper bounded by $|V|=2^{p-1}$. 
    Note that there always exists a simple cycle with minimum mean weight among all cycles, and the length of the cycle is also upper bounded by $|V|=2^{p-1}$.
    Now, we can take the sum of weights of all removed cycles in the previous procedures. If any removed cycle has mean weight higher than the minimum mean weight, then the mean weight of all edges in all removed cycles is higher than the minimum mean weight. If cycles with strictly larger mean occupy a nonvanishing fraction of the walk, their excess cost persists as $n\rightarrow\infty$. More precisely, let $\lambda_*$ be the minimum cycle mean and let $r\leq|V|-1$ be the length of the remaining path. Since all edge weights are nonnegative, the walk weight is at least $(n-p-r)\lambda_*$. Repeating a minimizing cycle and truncating its last repetition gives weight $(n-p)\lambda_*+\OC(1)$. Dividing these two bounds by $n-p$ proves the claimed convergence. A finite number of more expensive edges can remain without changing the limit.
\end{proof}

With the above lemma, the question we aim to solve reduces to finding the minimum-mean-weight cycle of the de Bruijn graph. An asymptotically optimal ordering pattern $c$ is the repetition of a minimum-mean-weight cycle.
If the cycle contains $m$ edges, then the resulting $c$ corresponds to $m$-translation-invariant permutations (see \cref{def:tipermutation}).

Karp's algorithm finds such a cycle in $\OC(|V||E|)=\OC(2^{2p})$ time (for the general $q$-bound, this is $\OC(2^{2p+q+v})$)~\cite{karp1978characterization}. This cost is manageable for the formula orders of interest ($2^p, 2^q \ll n$).

We test this strategy based on the leading-order error bound for an ensemble of $500$ random translation-invariant 1D nearest-neighbor Hamiltonians for fourth- and sixth-order product formulas. 
We find that sequential circuits, corresponding to $c=00\dots0$ or $c=11\dots1$, minimize the mean local predictor for almost all sampled Hamiltonians at both orders. There are very few counterexamples, one of which is the TFIM shown in \cref{fig_illu}.

Our finding that sequential circuits generally perform well is intuitively plausible but is not a rigorously proved theorem. However, in \cref{sec:nondegenerate}, we prove that for all translation-invariant 1D nearest-neighbor Hamiltonians, the estimated nondegenerate contribution to the leading-order error at orders two and four is smaller for the sequential ordering than for the brickwall ordering. This also explains why TFIM is an exception: its leading-order error has no nondegenerate contribution and consists of only the degenerate part.

\subsection{Optimal bounded-depth permutation}
\begin{figure}[t!]
    \centering
    \includegraphics[width=0.7\linewidth]{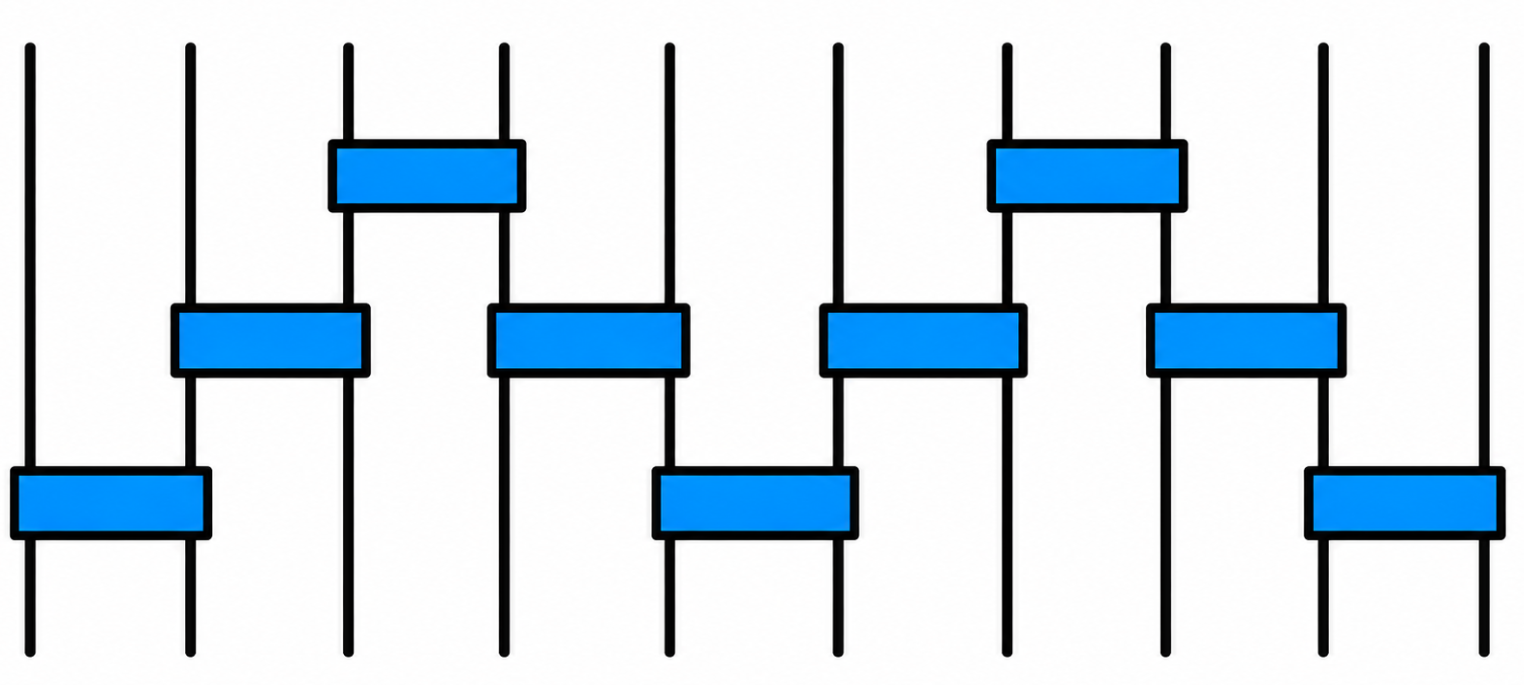}
    \caption{A wave-like parallel sequential circuit with depth $3$. Note the difference between this and the original parallel-sequential circuit introduced in \cref{fig_illu}.
    }
    \label{fig:wavelike}
\end{figure}

The sequential circuit, selected as the best ordering for most Hamiltonians in our tests, has depth $n$ for every first-order stage of the product formula, which may not be practical for near-term digital quantum simulation.
Therefore, finding the optimal permutation with a given upper bound on the depth per stage is also of interest. Now, we solve this problem using a modified de Bruijn graph.
We first show that the circuit depth can be easily read from the consecutive-label ordering pattern $c$.
\begin{lemma}
    The minimum circuit depth of a single stage in the product formula $S_c^{(p)}(\tau)$ is $1$ plus the length of the longest all-0 or all-1 substring of $c$, i.e.,
    \begin{equation}
    \begin{aligned}
       & \quad \, \, \mathrm{depth}(c)  \\
        &=1 + \max_{1\leq j\leq l\leq n-1}\{l-j+1: c_j=c_{j+1}=\dots=c_l\}.
    \end{aligned}
    \end{equation}
\end{lemma}
\begin{proof}
    If $c_j=c_{j+1}=\dots=c_{j+k}=1$, then $e^{-iH_j\delta_j},e^{-iH_{j+1}\delta_{j+1}},\dots,e^{-iH_{j+k+1}\delta_{j+k+1}}$ must appear sequentially in one stage, and since adjacent terms have overlapping supports, they cannot be parallelized. The same holds for the all-$0$ case. 
    Conversely, orient the path of Hamiltonian terms according to $c$, and assign each gate to a layer equal to one plus the length of the longest directed path ending at that gate. Adjacent gates then occupy different layers in the required order. The longest directed path contains one more vertex than the longest uniform substring of $c$, so this schedule attains the stated depth.
\end{proof}

Therefore, if the depth $d$ of a single stage is given, one must ensure the consecutive-label order $c$ does not contain any all-$0$ or all-$1$ substring of length $\geq d$. We can modify the de Bruijn graph to reflect these restrictions.
\begin{itemize}
    \item If $d\leq p$, the vertices and edges remain unchanged, but the weights should be
    \begin{equation*}
    w(s,t) =
    \begin{cases}
f_{H,p}(s_1,s_2,\dots,s_{p-1},t_{p-1}), \, \text{if longest} \\
\quad \text{uniform substring of } (s_1,\dots,t_{p-1})\\
\quad \text{is shorter than } d,\\
\infty,\quad  \text{otherwise}.
\end{cases}
\end{equation*}
    \item If $d>p$, the de Bruijn graph now has vertex set $V = \{0,1\}^{d-1}$ and edge set
    \begin{equation}
    \begin{aligned}
        E &=  \{((v_1,v_2,\dots,v_{d-1}),(v_2,v_3,\dots,v_{d-1},x)):\\
        &\quad \quad x\in\{0,1\}, v_i\in\{0,1\}, 1\leq i\leq d-1 \}
    \end{aligned}
    \end{equation}
    with weight function
     \begin{equation*}
    w(s,t) =
    \begin{cases}
|f_{H,p}(s_{d-p+1},\dots,s_{d-1},t_{d-1})|, \quad \text{if longest} \\ \quad \quad \text{uniform substring of } (s_1,\dots,t_{d-1}) \\ \quad \quad \text{is shorter than $d$},\\
\infty  \quad \text{otherwise}.
\end{cases}
\end{equation*}
\end{itemize}
Now, the optimal depth-$d$ consecutive-label ordering pattern $c$ can be obtained by finding the minimum mean weight cycle of the modified de Bruijn graph. Karp's algorithm runs in time $\OC(|V||E|)=\OC(4^{\max(p,d)})$, which is efficient when $d=\OC(\log n)$. 

We test this algorithm for $d=3$ on the same $500$ random Hamiltonians. At fourth order, all $500$ selected patterns are \emph{wave-like} (parallel-sequential) circuits (see \cref{fig:wavelike}), with pattern $c=0^{d-1}1^{d-1}0^{d-1}1^{d-1}\dots$ or $c=1^{d-1}0^{d-1}1^{d-1}0^{d-1}\dots$, including translations of these patterns. At sixth order, $466$ are wave-like and $34$ ($6.8\%$) instead have the period-three pattern $110110\dots$, which is also a type of parallel-sequential circuit. We present their actual-error performance in \cref{fig:karp_performance}.

\section{Numerical validation}
\label{sec:numerics}
We numerically test the effectiveness of our error bounds, in terms of tightness and success fraction in predicting which ordering is better. We also test the performance of the term ordering patterns selected by the de Bruijn graph-based algorithm among the $2^{n-1}$ candidates.

\subsection{Ordering-aware Trotter error bounds}
\label{sec:num_bounds}

In this subsection, we test our error bounds for
two characteristic permutations: brickwall (BW) and sequential (SEQ). Note that both SEQ and BW have two variants: forward SEQ ($c=11\dots 1$), backward SEQ ($c=00\dots 0$), even-odd BW ($c=0101\dots$), and odd-even BW ($c=1010\dots$).
Therefore, in the small-system comparisons of \cref{fig:comparison_1}, we average over the two variants of each circuit family. We also compute the \emph{BW/SEQ ratio}, the ratio between BW and SEQ's averaged Trotter error or upper bounds. The larger-system tests use the representative variants specified in the figure captions.

We start with small-scale numerical tests, presented in \cref{fig:comparison_1}, and then scale up to many qubits, as shown in \cref{fig:scaleup,fig:stepnumber,fig:stepnumber_6th}. We introduce the settings as follows.

\subsubsection{Small-scale validations}

We test our leading-order estimate and $q=2$ bound on $10$-qubit systems, and compare them with the state-of-the-art Trotter error bounds in Refs.~\cite{childs2021theory,schubert2023trotter}.
We use step size $\tau=0.1$ and Suzuki formulas with orders $p=2,4,6$. We consider three fixed translation-invariant models defined in \cref{tab:hamiltonians} and $500$ random translation-invariant Hamiltonians. 

\begingroup
\renewcommand{\arraystretch}{1.3}
\begin{table}[h!]
\centering
\begin{tabular}{c|l}
\hline
Model & Unnormalized local term $\tilde h_j$ \\
\hline
Heisenberg & $X_jX_{j+1}+Y_jY_{j+1}+Z_jZ_{j+1}$ \\
TFIM & $X_jX_{j+1}+Z_j$ \\
XYZ & $X_jX_{j+1}+0.7Y_jY_{j+1}+1.3Z_jZ_{j+1}+0.4Z_j$ \\
\hline
\end{tabular}
\caption{Fixed models used in \cref{fig:comparison_1}, with $h_j=\tilde h_j/\|\tilde h_j\|$ and $H=\sum_{j=1}^{N-1}h_j$. On-site fields act on the left site of each bond.}
\label{tab:hamiltonians}
\end{table}
\endgroup

We generate the random Hamiltonians through the following procedure: we draw a complex $4\times4$ matrix $G$ whose real and imaginary entries are independent standard Gaussian variables, and normalize it as
\begin{equation}
h=\frac{G+G^{\dagger}}{\|G+G^{\dagger}\|},
\qquad H=\sum_{j=1}^{N-1}h_j,
\label{eq:num_random_ti}
\end{equation}
where $h_j$ is the same $h$ translated to qubits $j$ and $j+1$.

\begin{figure*}[t!]
\centering
\includegraphics[width=\textwidth]{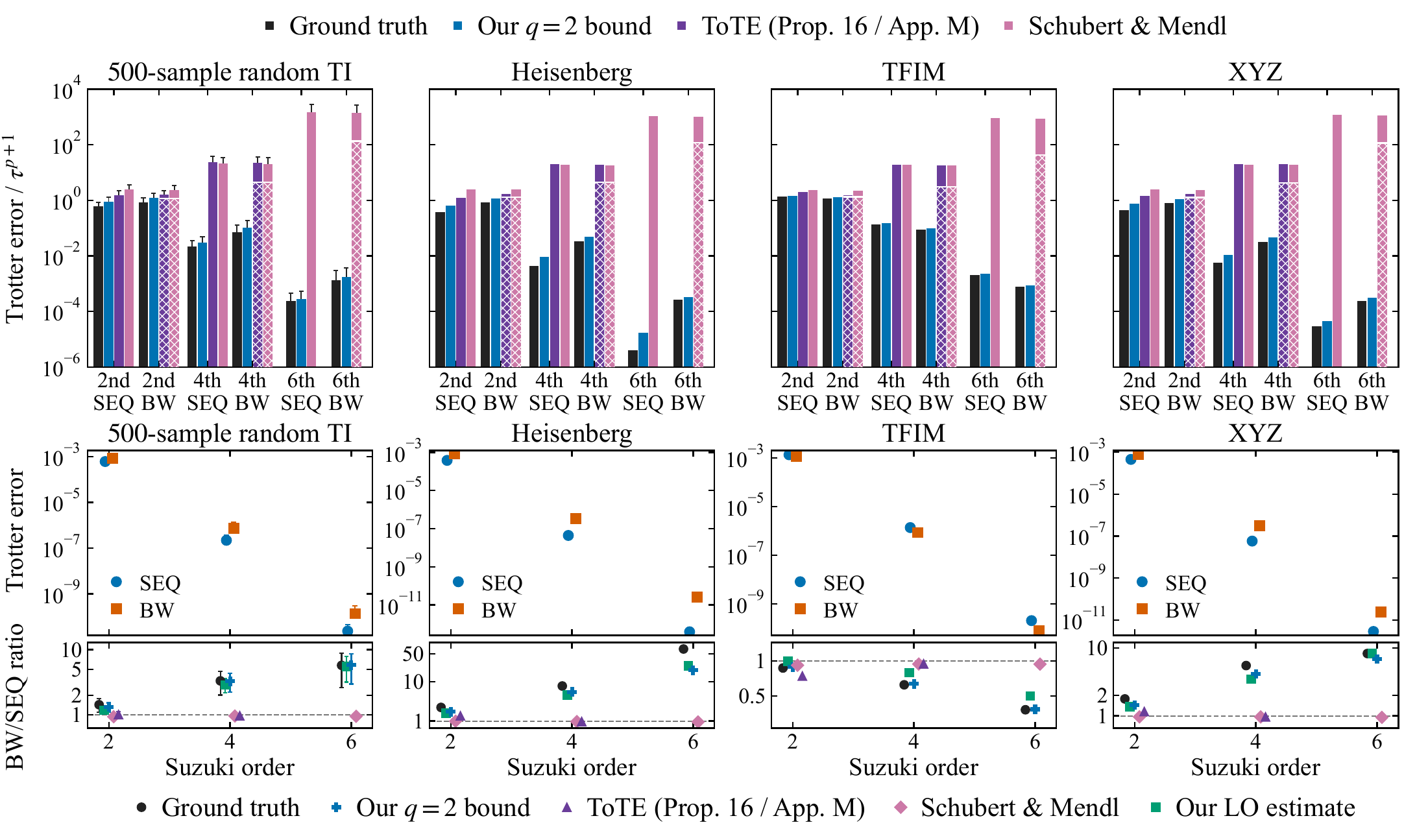}
\caption{Single-step errors and bounds for $10$-qubit open chains at $\tau=0.1$
Columns show the $500$-sample random translation-invariant ensemble and three fixed models. Top panels show actual errors and full bounds, both
divided by $\tau^{p+1}$. 
Note that the ToTE and \textcite{schubert2023trotter} BW bars have two layers: the full heights correspond to individual-term bounds; the shaded heights show the tighter bounds based on the grouping technique. 
ToTE uses Proposition 16 at $p=2$ and Appendix M at $p=4$. Lower panels show actual errors without $\tau^{p+1}$ rescaling
and BW/SEQ ratios, including the LO predictions. All ratio comparisons use individual-term bounds for BW, as the grouping version only works for BW and tightens the BW bound only. Random-ensemble values are means after direction averaging; ratio means are formed sample by sample. Error bars are sample standard deviations, shown only upward in error panels. 
}
\label{fig:comparison_1}
\end{figure*}

We compute the Trotter error bounds for BW and SEQ and the BW/SEQ ratios for the three fixed Hamiltonians and the $500$ random samples (for the random Hamiltonians, we also compute their means and standard deviations) in the following settings.
\begin{itemize}
    \item \emph{Ground truth}: we evaluate the exact Trotter error for BW and SEQ by directly computing $S^{(p)}_c(\tau)-e^{-iH\tau}$ and taking its spectral norm.
    \item \emph{Leading-order error estimate (LO)}: we  use the error bound introduced in \cref{eq:newbound,eq:fhp} to compute the leading-order error estimate.
\item \emph{$q=2$ Trotter error bound}: we use the $q=2$ (and $v=0$) construction of the full Trotter error bound in \cref{eq:q2}, which adds two bonds at each increment. At second order, $L=p+q=4$, and we let the consecutive-label ordering pattern grow through $3\to5\to7\to9$ terms; at fourth order, we use $L=6$ and the pattern grows through $5\to7\to9$ terms. 
However, one caveat here is that, at sixth order, with $L=p+q=8$, the permutation pattern only grows once from $7$ terms to $9$ terms, and the nonlocal tail computation is involved only once. This means that the small-scale numerics for the $q=2$ bound at $p\geq 6$ are similar to directly computing the ground truth. Therefore, the true performance and advantage of our bound for $p\geq 6$ product formulas is better demonstrated in the large-scale error estimation shown in \cref{fig:scaleup}, and the Trotter step estimation in \cref{fig:stepnumber,fig:stepnumber_6th}.

\item \emph{ToTE bound}: We compute the small-prefactor bounds of Ref.~\cite{childs2021theory}, abbreviated ToTE (Theory of Trotter Error).
At second order, we use Proposition 16 of \cite{childs2021theory}, which bounds sums of double commutators while keeping some matrix sums inside the norms. At fourth order, we use the numerical bound introduced in Appendix M of \cite{childs2021theory}, which utilizes a weighted sum of nested-commutator norms. For SEQ, these nested commutators are generated according to the Suzuki formula that contains $n$ local gates at every stage. For BW, there is an alternative approach to utilize the ToTE bound: one can choose to group all commuting terms as $H_\mathrm{even}$ and $H_\mathrm{odd}$ such that $H=H_\mathrm{even} + H_\mathrm{odd}$. Using this grouping strategy, instead of considering the BW circuit as $e^{-i H_{1} \tau} e^{-i H_{3} \tau} \cdots e^{-i H_{2} \tau} e^{-i H_{4} \tau}$, one can simply write the circuit as $e^{-i H_\mathrm{odd} \tau} e^{-i H_\mathrm{even} \tau}$, which simplifies the expressions for nested commutators and numerically gives tighter bounds for BW. 
However, directly evaluating these grouped commutator norms using full-chain matrices is not scalable to larger $n$. A scalable bound can instead be obtained by exploiting geometric locality and applying additional triangle inequalities, as in the large-scale tests below.
In our small-scale numerics, we present the bounds with both decompositions for a fair comparison. We restrict the ToTE comparison to its explicitly evaluated second- and fourth-order small-prefactor bounds. At sixth order, we only use the general higher-order extension of this approach given by Schubert and Mendl, described below.

\item \emph{\textcite{schubert2023trotter} bound}: we test the small-prefactor bounds of Theorem 1 of \textcite{schubert2023trotter}, which is also based on computing nested commutators, but keeps additional sums inside norms and applies at arbitrary order. We use the standard ``central splitting point" version of their bound, without optimizing over that choice. 
Note that its evaluation is computationally demanding at higher orders: the $n=9$ sixth-order formula has $401$ exponential stages after adjacent equal generators are merged, giving many commutator contributions to process. For $M$ individual $2$-qubit evolutions ($M=\mathcal{O}(n5^{p/2})$ for the Suzuki formula in \cref{eq:recursive_suzuki}), directly enumerating the multi-indices 
arising in this approach
involves $\mathcal{O}(M^p)$ commutator-norm terms at fixed $p$. In order to reduce the computational cost, we take advantage of geometric locality, combine repeated contributions by dynamic programming, and reuse local commutator norms across translated terms. For translation-invariant chains at fixed $p$, the scaling in $n$ can be reduced to $\mathcal{O}(n^2)$,
but the dependence on $p$ is still costly: at sixth order, the four ten-qubit circuit variants require computing $45{,}737$ different nested-commutator norms of matrices on at most eight qubits, even after combining commuting adjoint operations.
The \textcite{schubert2023trotter} bound also supports a less efficient but tighter grouping strategy for BW circuits, so we present the results both with and without grouping for BW.

\end{itemize}

We show the results of these small-scale numerical tests in \cref{fig:comparison_1} and briefly summarize them as follows.

The ground truth is that sequential ordering is better for Heisenberg and XYZ at all three orders, whereas brickwall is better for TFIM. The $q=2$ bound correctly predicts the preferred ordering at second and fourth order. For example, at fourth order, the actual and predicted ratios are $7.72$ and $5.44$ for Heisenberg, and $0.624$ and $0.636$ for TFIM.
The LO bound predicts the preferred ordering except for second-order TFIM, where it gives a tie. By contrast, the individual-bond reference bounds favor BW at fourth order for all three models.

For the ensemble of $500$ random translation-invariant Hamiltonians, the ground truth is that SEQ has smaller Trotter error in $446$, $493$, and $497$ of the $500$ samples at orders two, four, and six, respectively; BW is better in the remaining $54$, $7$, and $3$ samples.

At fourth order, the Schubert \& Mendl and ToTE bounds have mean bound-to-error ratios of about $1088$ and $1195$ for SEQ, and $342$ and $383$ for BW, respectively. 
Grouping the BW terms tightens these reference bounds, but a substantial gap remains. 
At fourth order, they predict mean BW/SEQ ratios of $0.959$ and $0.967$, respectively, whereas the actual mean ratio is $3.32$. Both favor BW in every sample, giving a success rate of only $1.4\%$. At second order, the success rates are $10.8\%$ for \textcite{schubert2023trotter} and $66.6\%$ for ToTE.

The mean bound-to-error ratios over the $500$ random Hamiltonians predicted by our $q=2$ bound for SEQ and BW are, respectively, $1.51$ and $1.41$ at second order, $1.39$ and $1.40$ at fourth order, and $1.22$ and $1.26$ at sixth order. The $q=2$ bound predicts the better ordering out of BW and SEQ in $481/500$, $495/500$, and $499/500$ samples at orders two, four, and six, respectively. 
However, we emphasize again that our $q=2$ bound for sixth-order product formulas uses slices containing $L=8$ terms. For a small system of only $10$ qubits, evaluating our bound is therefore already comparable in cost to computing the actual Trotter error directly. To better assess the practical performance of our bounds, we refer the reader to the following \emph{large-scale tests}, which consider systems of up to $256$ qubits.

\subsubsection{Larger-scale tests}

In the small-scale tests, the slices used in the $q=2$ calculation cover much of the chain, and some increments can be evaluated on complete subchains. Therefore, we must test the bound when the system is much larger than the slice to estimate how much improvement our new bounds can achieve.

We extend the three fixed models to $12$--$256$ qubits, keeping $\tau=0.1$ and norm-one local terms. We use $L=p+q=6$ for fourth order and $L=8$ for sixth order.
We compare with the \textcite{schubert2023trotter} bound, which applies at both orders and is slightly tighter than the fourth-order ToTE bound, as shown in \cref{fig:comparison_1}.

\begin{figure}[t]
    \centering
    \includegraphics[width=\linewidth]{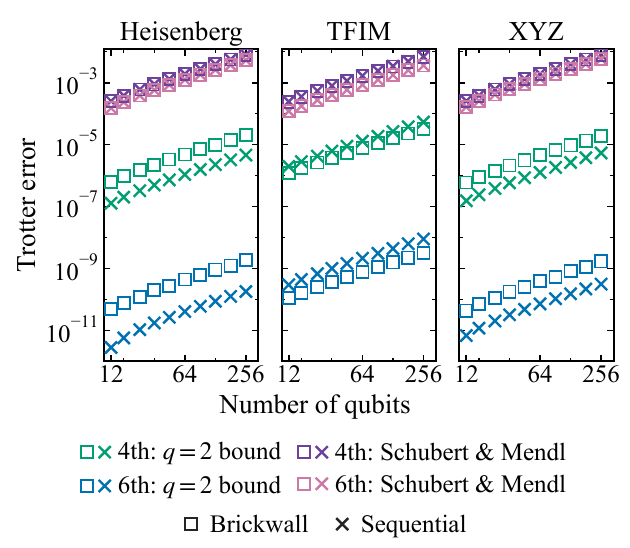}
    \caption{Fourth- and sixth-order Trotter error bounds for the Heisenberg, TFIM, and XYZ models on open chains of $12$--$256$ qubits, with $\tau=0.1$ and normalized local terms. The $q=2$ bound uses a slice of $L=6$ terms for $p=4$ and $L=8$ for $p=6$.
    \textcite{schubert2023trotter} uses the individual-term bound for BW. Colors distinguish the bound and formula order; open squares show even-odd BW and crosses show forward SEQ.}
    \label{fig:scaleup}
\end{figure}

As shown in \cref{fig:scaleup}, the $q=2$ bound remains multiple orders of magnitude tighter than the \textcite{schubert2023trotter} bound throughout this range. 
The $q=2$ bound continues to favor SEQ for Heisenberg and XYZ and BW for TFIM, consistent with the small-system ground truth, while the Schubert--Mendl values for the two circuits are nearly indistinguishable. Thus the improvement persists when the local calculation of our $q=2$ bound covers only a small part of the full system.

To see the actual improvement in resource estimation of quantum simulation, we also 
test how the improved error bound reduces the number of Trotter steps needed for a longer evolution. We first follow exactly the same settings in Ref.~\cite{childs2021theory}: in this test, we use the unnormalized random-field Heisenberg model,
\begin{equation}
H=\sum_{j=1}^{n}
\bigl(X_jX_{j+1}+Y_jY_{j+1}+Z_jZ_{j+1}+z_jZ_j\bigr),
\label{eq:num_random_field}
\end{equation}
where the $z_j$ are independent and uniformly random on $[-1,1]$. As in Ref.~\cite{childs2021theory}, we use five random-field realizations. For every realization, we first generate the longest $256$-qubit chain, and take their first $n$ fields for shorter chains. 
Just as in Ref.~\cite{childs2021theory}, we use the $4$th-order Suzuki formula, set the total evolution time as $t=n+1$, 
and let the target spectral-norm error be $10^{-3}$.

We also compute the number of Trotter steps required for $6$th-order Suzuki formula based on our $q=2,v=0$ bound, $q=2,v=1$ bound, and the \textcite{schubert2023trotter} bound.
We use the unnormalized Heisenberg model without the random field, because translation invariance allows us to reuse local commutator norms when evaluating the \textcite{schubert2023trotter} bound, substantially reducing the computational cost.

We compute the required step counts in five ways:
\begin{itemize}
    \item \emph{Ground truth}: we simulate the complete evolution and vary the number of Trotter steps until its error is below the target. This directly includes cancellation between steps. We evaluate chains of $4,6,8,10,12$ qubits.
    \item \emph{ToTE bound (4th-order only)}: we multiply the single-step bound by the number of steps and find a count that meets the target. We use the individual-bond  construction \cite[Appendix M]{childs2021theory} for SEQ and the two-layer construction for BW. For BW, we further expand the innermost layer into individual bonds before taking norms, so each norm can be evaluated on a finite local region. The time dependence is a $\tau^5$ factor, so the commutator norms need only be computed once for all different $\tau$s.
    \item \emph{\textcite{schubert2023trotter} bound (6th-order only)}: we use the same accumulation of single-step bounds. Again, we use \cite[Theorem 1]{schubert2023trotter} with the central splitting point and the individual-bond decomposition for both SEQ and BW, keeping the internal commutator sums inside the norms. At sixth order, this bound has a $\tau^7$ factor, so the commutators need only be computed once.
    \item \emph{$q=2$ bound for 4th order}: we use the same accumulation of single-step bounds with $L=6$.
    Unlike ToTE, this bound must be recomputed at each trial step size. We search over step counts, setting the step size to the total time divided by the count. The search has a resolution of $1\%$ or one step, and we report a count whose bound meets the target.
    \item \emph{$q=2$ bound and $q=2,v=1$ bound for $6$th order}: for the translation-invariant Heisenberg model, we test our bounds for $6$th-order Suzuki formulas. To get even tighter bounds, we also test the $(q=2,v=1)$ version with $L=9$ in addition to the standard $q=2$ bound with $L=8$.   
\end{itemize}

\begin{figure}[t]
\centering
\includegraphics[width=\linewidth]{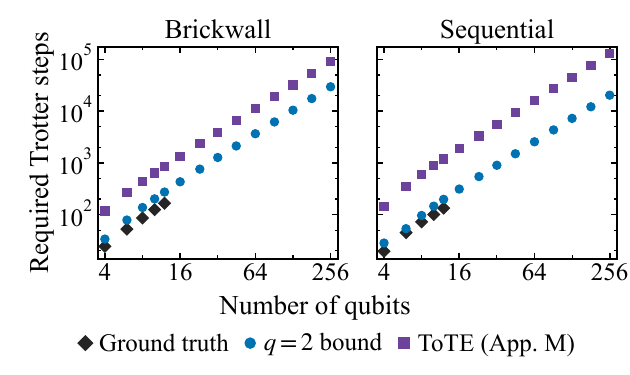}
\caption{Number of fourth-order Trotter steps for the random-field Heisenberg model in \cref{eq:num_random_field}, with $t=n+1$ 
and target error $10^{-3}$. Left: even--odd BW; right: forward SEQ. Points and error bars give means and sample standard deviations over five field realizations. Ground truth uses the complete multi-step evolution on $4,6,8,10,12$ qubits. The bound-based results extend to $256$ qubits and multiply a single-step bound by the number of steps. ToTE uses \cite[Appendix M]{childs2021theory}; the $q=2$ bound uses $L=6$.
}
\label{fig:stepnumber}
\end{figure}

\begin{figure}
    \centering    \includegraphics[width=\linewidth]{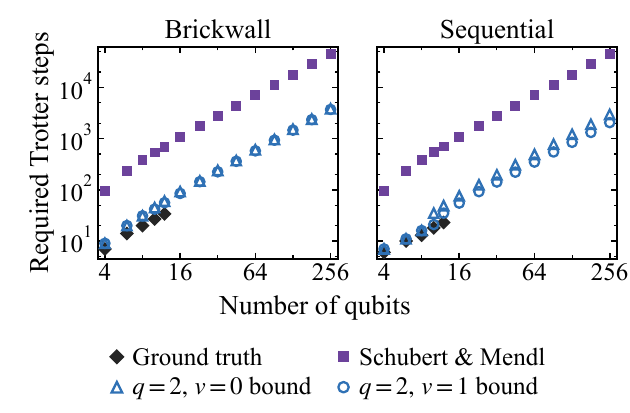}
    \caption{Number of sixth-order Trotter steps for the Heisenberg model, with $t=n+1$ and target error $10^{-3}$. Left: even--odd BW; right: forward SEQ. Ground truth uses the complete multi-step evolution on $4,6,8,10,12$ qubits. The bound-based results extend to $256$ qubits and multiply a single-step bound by the number of steps. We test two versions of our Trotter error bound: the $q=2, v=0$ bound and the $q=2, v=1$ bound. The latter is tighter and predicts fewer steps to achieve the same Trotter error.}
    \label{fig:stepnumber_6th}
\end{figure}

We show the results in \cref{fig:stepnumber}. For $12$ qubits, the mean counts for ground truth, 
$q=2$, and ToTE are approximately $166$, $273$, and $857$ for BW, and $133$, $197$, and $1198$ for SEQ. At $256$ qubits, the $q=2$ estimate is lower than ToTE by factors of about $3.1$ and $6.3$ for BW and SEQ, which is caused by multiple orders of magnitude of overestimation by the ToTE bound compared with the $q=2$ bound. The remaining gap from ground truth is partly due to the triangle inequality used to accumulate single-step bounds, which ignores cancellation between steps.

We show the sixth-order results in \cref{fig:stepnumber_6th}. For $12$ qubits, the counts for ground truth, $q=2,v=0$ bound, $q=2,v=1$ bound, and \textcite{schubert2023trotter} are approximately $34$, $60$, $56$, and $700$ for BW, and $23$, $49$, $35$, and $701$ for SEQ. Thus both versions of our bound predict step counts much closer to ground truth, and increasing the slice size from $8$ to $9$ terms gives a further improvement, especially for SEQ.

At $256$ qubits, the $q=2,v=0$ bound predicts approximately $3.8\times10^3$ steps for BW and $3.0\times10^3$ for SEQ. The $q=2,v=1$ bound reduces these counts to approximately $3.7\times10^3$ and $2.1\times10^3$, respectively, whereas \textcite{schubert2023trotter} predicts about $4.5\times10^4$ steps for either circuit. The $q=2,v=1$ estimates are therefore lower than Schubert--Mendl by factors of about $12$ and $22$ for BW and SEQ. At fixed formula order and circuit ordering, these reductions in step count give the same reductions in total gate count. These results demonstrate that the improvement from our bound persists at sixth order and can reduce the predicted circuit size by more than an order of magnitude.

\subsection{De Bruijn graph-based ordering optimization}
\label{sec:num_opt}

We now test the permutations selected by our de Bruijn graph-based algorithm. We use the same $500$ random local interactions on $8$-qubit chains, with $\tau=0.1$ and $p=4,6$. The graph edges are assigned the finite-time local weights in \cref{eq:finite_local_predictor}, evaluated at the simulation step size.

We compare three strategies: ``Karp optimal'' uses a minimum-mean-weight cycle without a depth restriction; ``$d=3$ Karp optimal'' restricts the depth of each first-order sweep to three; and ``Even--odd brickwall'' always uses the fixed even-odd brickwall ordering. We repeat each selected cycle's bit pattern and truncate it to $6$ bits.
We then compute the actual error of the selected circuit. Note that this is a small-system demonstration, so the edge effects' contribution is not negligible, and the two variants of SEQ and BW may lead to different Trotter errors. 

To determine how close this error is to the optimum, we also evaluate the exact Trotter errors corresponding to all $2^6=64$ consecutive-label ordering patterns for each sample. 
We define the \emph{optimality ratio} as the selected circuit's actual error divided by the smallest actual error among all $64$ patterns. A ratio of one means that the selected circuit is optimal. The empirical cumulative distribution function (CDF)
in \cref{fig:karp_performance} gives the fraction of samples whose optimality ratio is at most the value on the horizontal axis. A curve reaching one closer to an optimality ratio of one therefore indicates better performance.

\begin{figure}[t]
\centering
\includegraphics[width=\linewidth]{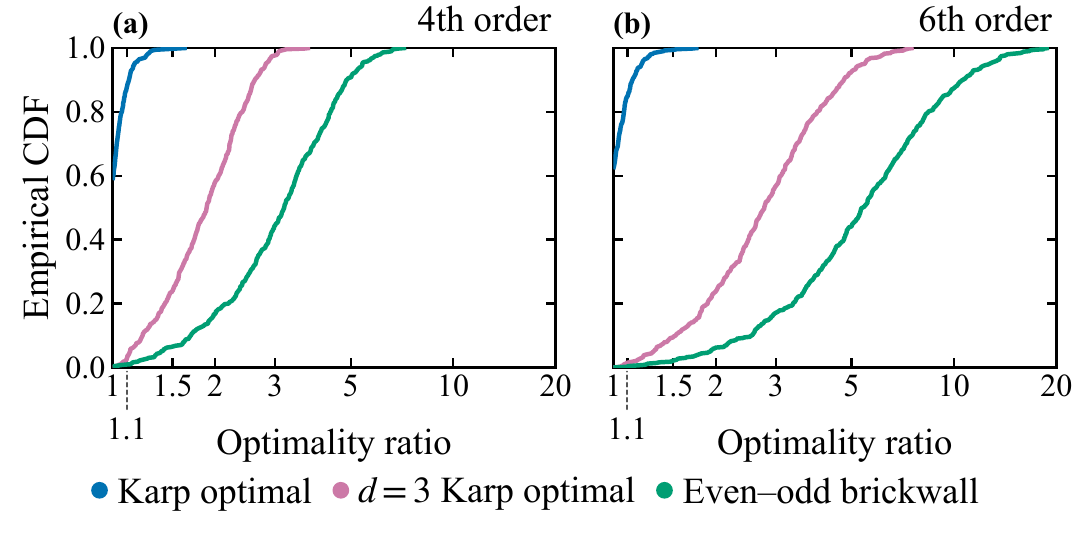}
\caption{Performance of unrestricted Karp optimization, depth-three Karp optimization, and fixed even--odd brickwall for $500$ random translation-invariant $8$-qubit chains at $\tau=0.1$. Panels (a) and (b) show fourth and sixth order. The optimality ratio is the actual error of the selected circuit divided by the minimum over all $64$ ordering patterns; the empirical CDF gives the fraction of samples at or below each ratio. Both Karp strategies use the local weights in \cref{eq:finite_local_predictor}.}
\label{fig:karp_performance}
\end{figure}

For all $500$ samples at both orders, unrestricted Karp optimization selects one of the two sequential directions. Its mean optimality ratios are $1.035$ and $1.042$ at fourth and sixth order, respectively, with $95$th percentiles $1.16$ and $1.21$. Thus the selected errors are only $3.5\%$ and $4.2\%$ above the optimum on average.

With depth restricted to three, the fourth-order optimizer selects translations of the wave-like pattern $0011\,0011\ldots$ in every sample. At sixth order, it selects this pattern in $466$ samples and translations of $110\,110\ldots$ in the remaining $34$. The mean optimality ratios are $1.92$ and $3.00$ at orders four and six, respectively, compared with $3.28$ and $6.10$ for brickwall. These results show that our local-weight optimization finds nearly optimal sequential orderings and also identifies more accurate shallow circuits than brickwall for this ensemble.

\section{Disucssions and outlook}
\label{sec:discuss}
In this work, for 1D nearest-neighbor systems, we develop a strongly ordering-aware Trotter error bound that can be efficiently evaluated on a classical computer. Next, we introduce the de Bruijn graph as a graph-theoretic representation of our error bound for Suzuki-Trotter formulas and propose an efficient classical algorithm to find the permutation minimizing our bound. Applying our strategy to random translation-invariant Hamiltonians suggests that sequential circuits have the optimal Trotter error performance for almost all cases tested in our study,
although other orderings, including brickwall, can perform best for some models. This agrees well with our intuition that deeper circuits can propagate correlations over longer distances, which may help reduce the Trotter error. For the sake of practicality, we optimize over modified de Bruijn graphs to obtain the 
best depth-restricted Trotter circuits. Our numerical results show that wave-like parallel-sequential circuits are frequently selected by the depth-restricted optimization.

One caveat of our Trotter error bounds is that they must be computed separately for each step size $\tau$, rather than simply rescaled by $\tau^{p+1}$. This dependence on $\tau$ is also one reason for their tightness: by exactly simulating the relevant local dynamics classically, our bounds capture all higher-order (and locally supported) contributions in $\tau$ beyond the leading-order term.

The ordering-aware error bounds extend directly to two-dimensional and other geometrically local systems.
When adding a new local term, we can simply include all previously added terms within a radius of order $p$ in the slice and take into account all interactions crossing the boundary of this region in the nonlocal correction. On a two-dimensional lattice, such a local region contains $O(p^2)$ qubits, so a direct implementation uses $2^{O(p^2)}$ classical computation for a local error increment. This cost remains independent of the total system size and is therefore constant for fixed product-formula order, although improving the dependence on $p$ is an interesting open problem.
Generalizing the permutation optimization to two and higher dimensions is a natural direction for future work. One challenge is to define sequential circuits in a meaningful way. Also, incorporating hardware-level considerations, including native gate sets, noise models, and fault-tolerant constraints, could better guide practical product-formula circuit design.

In brickwall implementations of Suzuki-Trotter formulas, some local evolution gates in the same second-order layer can be merged as evolutions of twice longer duration, suggesting brickwall circuits may need fewer gates to implement a single step and perform better in practice. However,
for fault-tolerant implementations,  recent results on angle-dependent Clifford+$T$ synthesis show that, in the small-angle regime, a rotation by angle $\theta$ can be synthesized to error $\delta$ using $\widetilde{O}(\theta^2/\delta)$ $T$ gates \cite{bothe2026more}. Consequently, under a fixed total synthesis-error budget, implementing two rotations of angle $\theta$ separately has essentially the same leading-order $T$-gate cost as implementing a single merged rotation by angle $2\theta$: assigning error $\delta/2$ to each separate rotation gives cost
\begin{equation}
2\,\widetilde{O}\left(\frac{\theta^2}{\delta/2}\right) = \widetilde{O}\left(\frac{(2\theta)^2}{\delta}\right).
\end{equation}
Thus, when the local evolutions are compiled into small-angle rotations, the gate-merging feature of brickwall circuits does not imply a corresponding advantage in $T$ count. A complete resource comparison between different orderings should therefore account for synthesis cost and circuit depth in addition to the number of local evolution gates.

Another direction is to design randomized product formulas informed by our results. Although Ref.~\cite{childs2019faster} found minimal advantage from fully randomized permutations for higher-order Suzuki formulas, their baseline was a variant of the low-depth ordering. Given the advantage of sequential circuits in many of our tests, a reasonable conjecture is that randomization over only variants of sequential circuits may lead to better performance than uniform randomness, and even better than deterministic sequential circuits. We leave these problems open for future research.

\section*{AI disclosure}
The research problem and the main conceptual ideas of this work were developed manually by the authors.
In particular, the authors derived the leading-order Trotter error bound in \cref{sec:leading}, developed the ordering-optimization algorithm in \cref{sec:optimization}, carried out extensive preliminary numerical experiments presented in \cref{sec:numerics}, and introduced the general framework of using local exact calculations to obtain Trotter error bounds for geometrically local systems in \cref{sec:geometry}. 
Generative AI, including ChatGPT 5.5, 5.6 Sol, and 6 Astra, subsequently derived a substantial portion of the technical details used to turn this framework into the full finite-step Trotter error bound presented in \cref{sec:geometry}. AI also assisted with later-stage numerical verification presented in \cref{sec:numerics} (based on code written by the authors), literature searches, and the writing and editing of the manuscript. All AI-generated arguments, calculations, numerical results, references, and text were independently checked and verified by the authors, who take full responsibility for the content and correctness of the paper.

\section*{Acknowledgments}

Z.L. and A.M.C. acknowledge support from the National Science Foundation (QLCI Grant No.~OMA-2120757). Z.-Y.W.~and A.V.G.~acknowledge support from the U.S.~Department of Energy, Office of Science, Accelerated Research in Quantum Computing, Fundamental Algorithmic Research toward Quantum Utility (FAR-Qu).
Z.-Y.W.~and A.V.G.~were also supported in part by NSF QLCI (awards No.~OMA-2120757 and No.~OSI-2553574), NQVL:QSTD:Design:FTL, DoE ASCR Quantum Testbed Pathfinder program (award No.~DE-SC0024220), ONR MURI,  NSF STAQ program, AFOSR MURI, and ARL (W911NF-24-2-0107). Z.-Y.W.~and A.V.G.~also acknowledge  support from the U.S.~Department of Energy, Office of Science, National Quantum Information Science Research Centers, Quantum Systems Accelerator (award No.~DE-SCL0000121).

\bibliography{ref.bib}

@article{childs2021theory,
  title = {{Theory of {T}rotter error with commutator scaling}},
  author = {Childs, Andrew M and Su, Yuan and Tran, Minh C and Wiebe, Nathan and Zhu, Shuchen},
  journal = {Physical Review X},
  volume = {11},
  number = {1},
  pages = {011020},
  year = {2021},
  publisher = {APS},
  doi = {10.1103/PhysRevX.11.011020}
}

@article{childs2019faster,
  title = {{Faster quantum simulation by randomization}},
  author = {Childs, Andrew M and Ostrander, Aaron and Su, Yuan},
  journal = {Quantum},
  volume = {3},
  pages = {182},
  year = {2019},
  publisher = {Verein zur F{\"o}rderung des Open Access Publizierens in den Quantenwissenschaften},
  doi = {10.22331/q-2019-09-02-182},
}

@article{karp1978characterization,
  title = {{A characterization of the minimum cycle mean in a digraph}},
  author = {Karp, Richard M},
  journal = {Discrete Mathematics},
  volume = {23},
  number = {3},
  pages = {309--311},
  year = {1978},
  publisher = {Elsevier},
  doi = {10.1016/0012-365X(78)90011-0}
}

@article{schubert2023trotter,
  title = {{Trotter error with commutator scaling for the {F}ermi-{H}ubbard model}},
  author = {Schubert, Ansgar and Mendl, Christian B},
  journal = {Physical Review B},
  volume = {108},
  number = {19},
  pages = {195105},
  year = {2023},
  publisher = {APS},
  doi = {10.1103/PhysRevB.108.195105},
}

@article{tranter2019ordering,
  title = {{Ordering of Trotterization: Impact on errors in quantum simulation of electronic structure}},
  author = {Tranter, Andrew and Love, Peter J and Mintert, Florian and Wiebe, Nathan and Coveney, Peter V},
  journal = {Entropy},
  volume = {21},
  number = {12},
  pages = {1218},
  year = {2019},
  publisher = {MDPI},
  doi = {10.3390/e21121218},
}

@article{childs2019nearly,
  title = {{Nearly optimal lattice simulation by product formulas}},
  author = {Childs, Andrew M and Su, Yuan},
  journal = {Physical Review Letters},
  volume = {123},
  number = {5},
  pages = {050503},
  year = {2019},
  publisher = {APS},
  doi = {10.1103/PhysRevLett.123.050503},
}

@misc{lane2025ordering,
  title = {{Ordering Matters: Structure, Accuracy and Gate Cost in Second-Order Suzuki Product Formulas}},
  author = {Lane, Matthew A and Browne, Dan E},
  jear = {2025},
  eprint = {2505.04240},
  archivePrefix = {arXiv},
  primaryClass = {quant-ph},
}

@article{huyghebaert1990product,
  title = {{Product formula methods for time-dependent {Schr{\"o}dinger} problems}},
  author = {Huyghebaert, J. and De Raedt, H.},
  journal = {Journal of Physics A: Mathematical and General},
  volume = {23},
  number = {24},
  pages = {5777--5793},
  year = {1990},
  doi = {10.1088/0305-4470/23/24/019}
}

@article{jahnke2000error,
  title = {{Error bounds for exponential operator splittings}},
  author = {Jahnke, Tobias and Lubich, Christian},
  journal = {BIT Numerical Mathematics},
  volume = {40},
  number = {4},
  pages = {735--744},
  year = {2000},
  doi = {10.1023/A:1022396519656}
}

@article{thalhammer2008high,
  title = {{High-order exponential operator splitting methods for time-dependent {Schr{\"o}dinger} equations}},
  author = {Thalhammer, Mechthild},
  journal = {SIAM Journal on Numerical Analysis},
  volume = {46},
  number = {4},
  pages = {2022--2038},
  year = {2008},
  doi = {10.1137/060674636}
}

@article{descombes2010exact,
  title = {{An exact local error representation of exponential operator splitting methods for evolutionary problems and applications to linear {Schr{\"o}dinger} equations in the semi-classical regime}},
  author = {Descombes, St{\'e}phane and Thalhammer, Mechthild},
  journal = {BIT Numerical Mathematics},
  volume = {50},
  number = {4},
  pages = {729--749},
  year = {2010},
  doi = {10.1007/s10543-010-0282-4}
}

@article{auzinger2013defect,
  title = {{Defect-based local error estimators for splitting methods, with application to {Schr{\"o}dinger} equations. {Part II}. Higher-order methods for linear problems}},
  author = {Auzinger, Winfried and Koch, Othmar and Thalhammer, Mechthild},
  journal = {Journal of Computational and Applied Mathematics},
  volume = {255},
  pages = {384--403},
  year = {2014},
  doi = {10.1016/j.cam.2013.04.043}
}

@article{childs2018toward,
  title = {{Toward the first quantum simulation with quantum speedup}},
  author = {Childs, Andrew M. and Maslov, Dmitri and Nam, Yunseong and Ross, Neil J. and Su, Yuan},
  journal = {Proceedings of the National Academy of Sciences},
  volume = {115},
  number = {38},
  pages = {9456--9461},
  year = {2018},
  doi = {10.1073/pnas.1801723115},
}

@misc{blunt2025monte,
  title = {{A {Monte Carlo} approach to bound {Trotter} error}},
  author = {Blunt, Nick S. and Ivanov, Aleksei V. and Bay-Smidt, Andreas Juul},
  year = {2025},
  eprint = {2510.11621},
  archivePrefix = {arXiv},
  primaryClass = {quant-ph},
}

@misc{maxwell2026practical,
  title = {{Practical estimation of {Trotter} error for {Hamiltonian} simulation}},
  author = {Maxwell, William and Casares, Pablo A. M. and Lang, Robert A. and Fomichev, Stepan and Arrazola, Juan Miguel and Jahangiri, Soran and Asadi, Ali and Nunez Meneses, Luis Alfredo and Germain, Thomas and Motlagh, Danial},
  year = {2026},
  eprint = {2606.30738},
  archivePrefix = {arXiv},
  primaryClass = {quant-ph},
}

@article{barthel2012quasilocality,
  author = {Barthel, Thomas and Kliesch, Martin},
  title = {{Quasilocality and Efficient Simulation of {Markovian} Quantum Dynamics}},
  journal = {Physical Review Letters},
  volume = {108},
  number = {23},
  pages = {230504},
  year = {2012},
  doi = {10.1103/PhysRevLett.108.230504},
}

@article{suzuki1991general,
  title = {{General theory of fractal path integrals with applications to many-body theories and statistical physics}},
  author = {Suzuki, Masuo},
  journal = {Journal of Mathematical Physics},
  volume = {32},
  number = {2},
  pages = {400--407},
  year = {1991},
  doi = {10.1063/1.529425}
}

@article{georgescu2014quantum,
  author = {Georgescu, I M and Ashhab, S and Nori, Franco},
  journal = {Reviews of Modern Physics},
  number = {1},
  pages = {153--185},
  publisher = {APS},
  title = {{Quantum simulation}},
  volume = {86},
  year = {2014},
  doi = {10.1103/RevModPhys.86.153},
}

@article{RevModPhys.83.863,
  author = {Polkovnikov, Anatoli and Sengupta, Krishnendu and Silva, Alessandro and Vengalattore, Mukund},
  doi = {10.1103/RevModPhys.83.863},
  issue = {3},
  journal = {Reviews of Modern Physics},
  month = aug,
  numpages = {0},
  pages = {863--883},
  publisher = {American Physical Society},
  title = {{Colloquium: Nonequilibrium dynamics of closed interacting quantum systems}},
  volume = {83},
  year = {2011},
}

@article{trotter1959product,
  author = {Trotter, Hale F},
  journal = {Proceedings of the American Mathematical Society},
  number = {4},
  pages = {545--551},
  publisher = {JSTOR},
  title = {{On the product of semi-groups of operators}},
  volume = {10},
  year = {1959},
  doi = {10.1090/S0002-9939-1959-0108732-6}
}

@article{suzuki1990fractal,
  author = {Suzuki, Masuo},
  journal = {Physics Letters A},
  number = {6},
  pages = {319--323},
  publisher = {Elsevier},
  title = {{Fractal decomposition of exponential operators with applications to many-body theories and Monte Carlo simulations}},
  volume = {146},
  year = {1990},
  doi = {10.1016/0375-9601(90)90962-N}
}

@article{Lloyd1996,
  author = {Lloyd, S},
  doi = {10.1126/science.273.5278.1073},
  journal = {Science},
  month = aug,
  pages = {1073--1078},
  title = {{Universal Quantum Simulators}},
  volume = {273},
  year = {1996},
}

@article{berry2007efficient,
  author = {Berry, Dominic W and Ahokas, Graeme and Cleve, Richard and Sanders, Barry C},
  journal = {Communications in Mathematical Physics},
  pages = {359--371},
  publisher = {Springer},
  title = {{Efficient quantum algorithms for simulating sparse Hamiltonians}},
  volume = {270},
  year = {2007},
  doi = {10.1007/s00220-006-0150-x},
}

@article{Vidal2004,
  author = {Vidal, Guifr{\'{e}}},
  issn = {00319007},
  journal = {Physical Review Letters},
  number = {4},
  pages = {040502},
  title = {{Efficient simulation of one-dimensional quantum many-body systems}},
  volume = {93},
  year = {2004},
  doi = {10.1103/PhysRevLett.93.040502},
}

@article{paeckel2019time,
  author = {Paeckel, Sebastian and K{\"{o}}hler, Thomas and Swoboda, Andreas and Manmana, Salvatore R and Schollw{\"{o}}ck, Ulrich and Hubig, Claudius},
  journal = {Annals of Physics},
  pages = {167998},
  publisher = {Elsevier},
  title = {{Time-evolution methods for matrix-product states}},
  volume = {411},
  year = {2019},
  doi = {10.1016/j.aop.2019.167998}
}

@article{wei2025state,
  author = {Wei, Zhi-Yuan and Malz, Daniel},
  doi = {10.22331/q-2026-04-21-2079},
  issn = {2521-327X},
  journal = {{Quantum}},
  month = apr,
  pages = {2079},
  publisher = {{Verein zur F{\"{o}}rderung des Open Access Publizierens in den Quantenwissenschaften}},
  title = {{State preparation with parallel-sequential circuits}},
  volume = {10},
  year = {2026},
}

@article{Schon2005,
  author = {Sch{\"{o}}n, C. and Solano, E. and Verstraete, F. and Cirac, J. I. and Wolf, M. M.},
  issn = {00319007},
  journal = {Physical Review Letters},
  number = {11},
  pages = {110503},
  title = {{Sequential generation of entangled multiqubit states}},
  volume = {95},
  year = {2005},
  doi = {10.1103/PhysRevLett.95.110503},
}

@article{Schon2007,
  author = {Sch{\"{o}}n, C. and Hammerer, K. and Wolf, M. M. and Cirac, J. I. and Solano, E.},
  issn = {10502947},
  journal = {Physical Review A},
  number = {3},
  pages = {032311},
  title = {{Sequential generation of matrix-product states in cavity QED}},
  volume = {75},
  year = {2007},
  doi = {10.1103/PhysRevA.75.032311}
}

@article{zypp,
  author = {Wei, Zhi Yuan and Malz, Daniel and Cirac, J Ignacio},
  issn = {10797114},
  journal = {Physical Review Letters},
  number = {1},
  pages = {010607},
  pmid = {35061477},
  title = {{Sequential Generation of Projected Entangled-Pair States}},
  volume = {128},
  year = {2022},
  doi = {10.1103/PhysRevLett.128.010607}
}

@article{Banuls2008,
  author = {Ba{\~{n}}uls, M. C. and P{\'{e}}rez-Garc{\'{i}}a, D. and Wolf, M. M. and Verstraete, F. and Cirac, J. I.},
  issn = {10502947},
  journal = {Physical Review A},
  number = {5},
  pages = {052306},
  title = {{Sequentially generated states for the study of two-dimensional systems}},
  volume = {77},
  year = {2008},
  doi = {10.1103/PhysRevA.77.052306},
}

@article{chen2024sequential,
  author = {Chen, Xie and Dua, Arpit and Hermele, Michael and Stephen, David T and Tantivasadakarn, Nathanan and Vanhove, Robijn and Zhao, Jing-Yu},
  journal = {Physical Review B},
  number = {7},
  pages = {075116},
  publisher = {APS},
  title = {{Sequential quantum circuits as maps between gapped phases}},
  volume = {109},
  year = {2024},
  doi = {10.1103/PhysRevB.109.075116},
}

@article{yj_circ,
  author = {Liu, Yu-Jie and Shtengel, Kirill and Smith, Adam and Pollmann, Frank},
  issue = {4},
  journal = {PRX Quantum},
  month = nov,
  numpages = {20},
  pages = {040315},
  publisher = {American Physical Society},
  title = {{Methods for Simulating String-Net States and Anyons on a Digital Quantum Computer}},
  volume = {3},
  year = {2022},
  doi = {10.1103/PRXQuantum.3.040315},
}

@article{Bravyi2006,
  author = {Bravyi, S. and Hastings, M. B. and Verstraete, F.},
  doi = {10.1103/PhysRevLett.97.050401},
  issn = {00319007},
  journal = {Physical Review Letters},
  number = {5},
  pages = {050401},
  title = {{Lieb-Robinson bounds and the generation of correlations and topological quantum order}},
  volume = {97},
  year = {2006},
}

@incollection{lieb1972finite,
  author = {Lieb, Elliott H and Robinson, Derek W},
  pages = {425--431},
  publisher = {Springer},
  title = {{The finite group velocity of quantum spin systems}},
  year = {1972},
  doi = {10.1007/978-3-662-10018-9_25},
  booktitle = {Statistical Mechanics},
  editor = {Nachtergaele, Bruno and Solovej, Jan Philip and Yngvason, Jakob},
  address = {Berlin, Heidelberg}
}

@misc{bothe2026more,
  title={{More efficient Clifford+ T synthesis for small-angle rotations and application to Trotterization}},
  author={Bothe, Marius and S{\"u}nderhauf, Christoph and Witham, Michael J and Campbell, Earl T and Blunt, Nick S},
  jear = {2026},
  eprint = {2605.31544},
  archivePrefix = {arXiv},
  primaryClass = {quant-ph},
}

\appendix
\newpage
\onecolumngrid

\crefalias{section}{appendix}

\newpage
\section{Nondegenerate leading-order error analysis}
\label{sec:nondegenerate}
In this appendix, we focus on the nondegenerate part of the leading-order Trotter error of the Suzuki-Trotter formula, as defined in \cref{sec:leading}. We use the de Bruijn graph-based optimization approach, as introduced in \cref{sec:optimization}, to give a rigorous proof based on numerical calculations for the following result.
\begin{theorem}
    For second- and fourth-order Suzuki–Trotter formulas applied to translation-invariant one-dimensional nearest-neighbor Hamiltonians, sequential orderings minimize the asymptotic mean per-term nondegenerate leading-order bound.
\end{theorem}

We use $\tilde{\mathcal{A}}_{p,c}^{\mathrm{ND}}$ to denote the nondegenerate contribution of the leading-order Trotter error:
\begin{equation}
\begin{aligned}
    \tilde{\mathcal{A}}_{p,c}^\mathrm{ND}  
    &\coloneqq -(-i)^{p+1}\sum_{j=1}^{n-p} \sum_{N\in \mathrm{NDNC}_j} \tau^{p+1} \epsilon_{p,c,j,N} N
\end{aligned}
\end{equation}
where 
$\epsilon_{p,c,j,N}\in \mathbb{R}$, with the common phase $-(-i)^{p+1}$ included so that the coefficients below are real, and
\begin{equation}
\begin{aligned}
    \mathrm{NDNC}_{p+1} (j,j+k) \coloneqq \{ [H_{j_1},\dots,[H_{j_{p}},H_{j_{p+1}}]]: (j_1,\dots,j_{p+1}) = \sigma(j,\dots,j+p), \sigma\in S_{p+1}\}
\end{aligned}
\end{equation}
is the set of all nondegenerate nested commutators of $H_j,H_{j+1},\dots,H_{j+p}$.

A key observation is the following lemma.
\begin{lemma}
    All nondegenerate nested commutators of $H_j,\dots,H_{j+p}$ are either $0$ or $\pm [H_{j},[H_{j+1},\dots,[H_{j+p-1}, H_{j+p}]]]$.
\end{lemma}
\begin{proof}
    We proceed by induction from $p=1$. It is obvious that $[H_{j+1}, H_j] = - [H_j, H_{j+1}]$. Next, for arbitrary $k$, we consider the process of assigning terms in $\{H_j,\dots, H_{j+k}\}$ into the nested commutator $[\cdot,[\cdot,\dots,[\cdot,\cdot]]]$ from the innermost layer to the outermost layer. Observe that the next term to be assigned must be connected to the already assigned terms to avoid non-overlapping support. Therefore, the outermost term (i.e., the first term in the nested commutator) must be either $H_j$ or $H_{j+k}$. This implies that there are only two types of nested commutators: $[H_j,N(j+1,\dots,j+k)]$ and $[H_{j+k},N(j,\dots,j+k-1)]$ where $N(j,\dots,j+k-1)$ is an arbitrary nondegenerate nested commutator of $H_j,\dots,H_{j+k-1}$ ($k$ terms in total). Assuming the lemma holds for $k$, every nonzero $N(j+1,\dots,j+k)$ is one of $\pm [H_{j+1},\dots,[H_{j+k-1}, H_{j+k}]]$. Therefore, every nonzero nondegenerate nested commutator of the second type is one of 
    \begin{equation}
    \begin{aligned}
        &\quad \pm[H_{j+k},[H_j,\dots,[H_{j+k-2},H_{j+k-1}]]] \\
        &= \mp [H_j, [H_{j+1},\dots,[H_{j+k-1}, H_{j+k}]]],
    \end{aligned}
    \end{equation}
    which is a nested commutator of the first type. Therefore, there is only one type of nondegenerate nested commutator for any $k$, and they are all equal to either $0$ or one of $\pm[H_j,\dots,[H_{j+k-1},H_{j+k}]]$.
\end{proof}

This lemma significantly reduces the degrees of freedom in the expression for $\tilde{\mathcal{A}}_{p,\pi}^{\mathrm{ND}}$:
\begin{equation}
    \tilde{\mathcal{A}}_{p,c}^{\mathrm{ND}} = -(-i)^{p+1}\sum_{j=1}^{n-p} \tau^{p+1}  \beta_{p,c,j} [H_j,\dots,[H_{j+p-1},H_{j+p}]].
\end{equation}
Note that here $\beta_{p,c,j}$ is the only quantity affected by the permutation. We further observe that only the relative order of $j,j+1,\dots,j+p$ matters, because other terms do not contribute to these nested commutators. Recall from \cref{def:relativeorder} that the relative order of these elements is given by a $p$-bit string $c' = (c_j,c_{j+1},\dots,c_{j+p-1})$. Now, we can write $\beta_{p,c,j}$ as a function of only $c'$: $\beta_{p,c,j} = f_p(c')$.

By the following triangle inequality, we have
\begin{equation}
\label{eq:triangle_general}
    \| \tilde{\mathcal{A}}_{p,c}^{\mathrm{ND}} \|  \leq  \tau^{p+1} \sum_{j=1}^{n-p} |\beta_{p,c,j}|   \left\|  [H_j,\dots,[H_{j+p-1},H_{j+p}]] \right\|.
\end{equation}
For the $j$th block, the nondegenerate nested commutator acts non-trivially on qubits $j,j+1,\dots,j+p+1$, and the nested commutators for blocks $j$ and $k$ generally do not commute or anticommute when $|j-k|\leq p+1$. Therefore, in general, the triangle inequality is not saturated. 
However, if a certain $c$ minimizes the upper bound on $\|  \tilde{\mathcal{A}}_{p,c}^{\mathrm{ND}} \|$, it may also lead to better practical performance. This is not rigorously guaranteed, but our numerics suggest that minimizing the bound is often a useful predictor of a low-error ordering, and for some Hamiltonians, the optimal ordering can be efficiently found.

In this appendix, we again focus on translation-invariant Hamiltonians because they are easier to analyze and arise in many scenarios of interest. A major simplification resulting from translation invariance is that all nondegenerate nested commutators are the same matrix with different supports. This means that there exists a real number $E_{H,p}$ such that
\begin{equation}
    \left\|  [H_j,\dots,[H_{j+p-1},H_{j+p}]] \right\| = E_{H,p}
\end{equation}
for all $1\leq j\leq n-p$. Therefore, the upper bound in \cref{eq:triangle_general} for translation-invariant Hamiltonians can be written as
\begin{equation}
\|\tilde{\mathcal{A}}_{p,c}^{\mathrm{ND}}\|\leq \tau^{p+1} E_{H,p} \sum_{j=1}^{n-p} |\beta_{p,c,j}|,
\end{equation}
which implies that minimizing this upper bound is equivalent to finding the consecutive-label ordering pattern $c$ that minimizes $\sum_j |\beta_{p,c,j}| = \sum_j |f_p(c_j,\dots,c_{j+p-1})|$. To solve this problem, we model the process of choosing $c$ as finding a minimum-mean-weight cycle in a de Bruijn graph, as shown in \cref{sec:optimization}.

We carry out brute-force computation of $f_p(c')$ values for the Suzuki-Trotter product formulas with $p\in\{2,4\}$. Surprisingly, we find that there are very few possible values of $f_p(c')$ when $p\in\{2,4\}$. Specifically, we find the following.
\begin{fact}
\label{fact:order2}
    $f_2(c') \in \{-1/6, 1/12\}$ 
\end{fact}
\begin{fact}
\label{fact:order4}
    $f_4(c') \in \{ -0.0009297, -0.0003951, 0.002650, -0.004509  \}$ (with precision up to 4 significant digits).
\end{fact}
We also find some rules to quickly decide the value of $f_4(c')$ for any given $c'$.
\begin{fact}
\label{fact:4values}
    $$f_4(c') =     \begin{cases}
-0.0009297 & \text{if $c'$ does not contain any $01$-substring}  \\
-0.0003951 & \text{if ($c'_1 c'_2 = 01$ and $c'_3 c'_4\neq 01$ or ($c'_1 c'_2 \neq 01$ and $c'_3 c'_4 = 01$)}\\
0.002650 & \text{if $c'_2 c'_3=01$}\\
-0.004509 & \text{if $c' = 0101$}.\\
\end{cases} $$
In other words, we can decide $f_4(c')$ by inspecting the number and position of the $01$-substring in $c'$.
\end{fact}

Running Karp's algorithm using these data for unbounded depth and depth $d$ gives the following results.
\begin{fact}
\label{fact:seq}
    For $p\in\{2,4\}$, a permutation minimizing the mean per-term bound on the nondegenerate leading-order error in 1D nearest-neighbor systems is the sequential permutation, i.e., $c=11\dots 1$ or $c=00\dots 0$.
\end{fact}
\begin{fact}
\label{fact:ps}
    For $p\in\{2,4\}$ and $d\geq2$, a depth-$d$ permutation minimizing this mean per-term bound is the $2(d-1)$-TI wave-like permutation, i.e., $c=1^{d-1} 0^{d-1} 1^{d-1}\dots$ or $c=0^{d-1} 1^{d-1} 0^{d-1}\dots$.
\end{fact}

The above facts constitute a computer-based proof of the optimality of sequential (for unbounded depth) and wave-like parallel-sequential (for bounded depth) permutations for the mean per-term bound on nondegenerate leading-order error contributions at orders two and four (though not for the complete finite-chain error).

\end{document}